\documentclass[conference]{IEEEtran}

\usepackage{cite}
\usepackage{amsmath}
\usepackage{amssymb}
\usepackage{amsfonts}
\usepackage{graphicx}
\usepackage{multicol}
\usepackage{multirow}
\usepackage{float}
\usepackage{textcomp}
\usepackage{csquotes}
\usepackage{lipsum}
\usepackage{easyReview}
\usepackage{tikz}
\usepackage{amsthm}
\usepackage{booktabs}
\usepackage{amssymb}
\usepackage{subcaption}
\usepackage{makecell}
\usepackage[ruled,linesnumbered,vlined]{algorithm2e}
\usepackage[capitalize,nameinlink]{cleveref}
\usetikzlibrary{shapes}
\usetikzlibrary{arrows.meta}
\usetikzlibrary{fit}
\usepackage[symbol]{footmisc}

\SetKwComment{Comment}{$\triangleright$~}{}
\SetKwProg{While}{while}{:}{}
\SetKwProg{For}{for}{:}{}
\SetKw{Continue}{continue}
\SetKwProg{Def}{def}{:}{}
\SetKwBlock{Begin}{Begin}{}

\tikzset{actor/.style={align=center, draw, rounded corners=.20cm}}
\tikzset{setback/.style={fill=gray!50}}
\tikzset{emphasize/.style={line width=2pt,font=\bfseries}}
\tikzset{virtualactor/.style={align=center, draw, ellipse}}
\tikzset{tapeactor/.style={tape,align=center,draw}}
\tikzset{decider/.style={align=center, draw, regular polygon, regular polygon sides=6}}
\tikzset{arc/.style={thin,-{Stealth}}}
\tikzset{arccontrol/.style={thin,-{Stealth},dashed}}

\def\BibTeX{{\rm B\kern-.05em{\sc i\kern-.025em b}\kern-.08emT\kern-.1667em\lower.7ex\hbox{E}\kern-.125emX}}

\newcommand{\myalert}[2]{\ifthenelse{\boolean{disablecomment}}{#1}{\alert{#1} \todo{\small #2}}}

\newtheorem{theorem}{Theorem}

\newtheorem{definition}{Definition}

\newtheorem{property}{Property}

\begin{document}

\title{Real-time Mode-Aware Dataflow: A Dataflow Model to Specify and Analyze Mode-dependent CPSs under Relaxed Timing Constraints}

\author{\IEEEauthorblockN{Guillaume Roumage$^\dagger$, Selma Azaiez$^\ddagger$, Cyril Faure$^\ddagger$, Stéphane Louise$^\ddagger$} \IEEEauthorblockA{$^\dagger$guillaume.roumage.research@proton.me \\ $^\ddagger$\textit{Université Paris-Saclay, CEA, List, F-91120, Palaiseau, France} \\ $^\ddagger$firstname.lastname@cea.fr}}

\maketitle

\begin{abstract}
  Modern Cyber-Physical Systems (CPS) often exhibit both relaxed real-time constraints and a mode-dependent execution. Relaxed real-time constraints mean that only a subset of the processes of a CPS have real-time constraints, and a mode-dependent CPS has conditional execution branches. Static analysis tools, such as the PolyGraph model (a formalism extending the Cyclo-Static Dataflow model with real-time constraints), can specify and analyze systems with relaxed real-time constraints. However, PolyGraph is limited in its ability to specify and analyze mode-dependent CPSs. This paper extends PolyGraph with routing actors, yielding the Routed PolyGraph model. This model is further extended to the Real-time Mode-Aware Dataflow (RMDF), which both leverages routing actors and incorporates a new dataflow actor to specify mode-dependent CPSs under relaxed real-time constraints. This paper also extends the static analyses of PolyGraph to RMDF. We showcase the application of RMDF with a specification and an analysis (derivation of timing constraints at the job-level and a feasibility test) of the vision processing system of the Ingenuity Mars helicopter.
\end{abstract}


\section{Introduction}

Cyber-Physical Systems (CPSs) are reactive systems that detect environmental shifts through sensors, process this information using computational processes, and then use the output to control actuators. CPSs range from digital signal processing systems to embedded/cloud infrastructures, soft/hard real-time systems, and even a mix of all the above. These complex systems must operate reliably without threatening their internal processes. For example, a failure of the actuator in an autonomous car can lead to catastrophic consequences such as a car crash or a pedestrian accident.

An increasing number of systems exhibit dynamic behavior, allowing them to switch between various modes of operation at runtime, e.g., a switch between the nominal and the degraded execution mode. The procedure that initiates these mode switches may rely on runtime data, limiting the static analyzability of the system. This paper presents an extension of an existing static analysis tool to specify and analyze such systems with dynamic behavior.

The design tools of interest in this paper are Dataflow Models of Computation and Communication (DF MoCC). These tools can formally prove the correctness of some classes of CPSs regarding safety and temporal properties, such as consistency, liveness, latency, and throughput. This paper presents a DF MoCC called Real-time Mode-aware Dataflow (RMDF), which extends the \emph{routed PolyGraph model}, this latter being itself an extension of the DF MoCC PolyGraph~\cite{dubrulle_data_2019}. While PolyGraph is tailored to specify and analyze CPSs with \emph{relaxed} real-time constraints, that is, a CPS with real-time constraints on only a subset of its processes, it cannot specify and analyze \emph{mode-dependent} CPSs.

\subsubsection{Contribution}

We have developed an extension of PolyGraph, called RMDF, to specify and analyze \emph{mode-dependent} CPSs under relaxed real-time constraints. As a dataflow model is a trade-off between its expressiveness and its analyzability~\cite{roumage_survey_2022}, the challenge here is to enhance the former without reducing the latter. Within the scope of this paper, a \emph{mode-dependent} CPS has a set of conditional execution branches. The consistency analysis (checking the existence of a memory-bounded execution) and the liveness analysis (checking the existence of a deadlock-free execution) of PolyGraph are extended to RMDF. In addition, the timing analysis described in~\cite{roumage_static_2024} for PolyGraph is extended to RMDF, and a feasibility test is proposed.

\subsubsection{Paper organization}

The paper starts by presenting the background and the terminology of the PolyGraph model in \cref{sec:background}. The PolyGraph model is extended to the \emph{routed PolyGraph model} in \cref{sec:routed-polygraph}, which is further extended to the RMDF model in \cref{sec:rmdf}. The extension of the static analysis of PolyGraph to RMDF is detailed in \cref{sec:static-analysis}. \cref{sec:timing-analysis} presents the timing analysis of an RMDF specification. This timing analysis permits a feasibility test presented in \cref{sec:feasibility-test}. \cref{sec:applications} applies the RMDF model to specify and analyze a dataflow specification of the Ingenuity vision processing system. \cref{sec:mode-change-protocol} discusses the mode change protocol of RMDF. Finally, \cref{sec:related-works} presents discussions and related works, and \cref{sec:conclusion-future-works} concludes the paper.

\section{Background and Terminology}

\label{sec:background}

RMDF is an extension of PolyGraph~\cite{dubrulle_polygraph_2021}, which is itself a superset of the Synchronous Dataflow (SDF) model~\cite{lee_synchronous_1987}. In this section, we briefly present SDF, and we explain how the PolyGraph model expands it. We also introduce \emph{disjoint directed meshes} to explain further how RMDF extends PolyGraph.

\subsection{The SDF Model}

An SDF specification is an oriented graph $G = (V, E)$ where $V$ is a finite set of \emph{actors} and $E$ is a finite set of \emph{channels}. An SDF specification is characterized by a \emph{topology matrix} $G_\Gamma$ of size $|E| \times |V|$. An \emph{actor} is a computational unit that both produces and consumes data tokens every time it is executed. The atomic amount of data exchanged is known as a \emph{token}. The entry $(i, j)$ of $G_\Gamma$ is the number of tokens produced or consumed by the actor $v_j$ on the channel $c_i$ each time the actor $v_j$ is executed (this number is positive if the tokens are produced, and negative otherwise). An execution of an actor is called a \emph{job}.

An actor $v \in V$ has (optional) input and output ports connected to input and output channels of $E$. A channel $c_i = (v_j, v_k, n_{ij}, n_{ik}, init_i) \in E$ connects an output port of the actor $v_j \in V$ to an input port of the actor $v_k \in V$. This channel also has a production rate $n_{ij} \in \mathbb{N}^*$ (which is the entry $(i, j)$ of $G_\Gamma$), a consumption rate $n_{ik} \in \mathbb{N}^*$ (the entry $(i, k)$ of $G_\Gamma$), and a number of initial tokens $init_i \in \mathbb{N}$. An actor's job produces/consumes tokens to/from the buffer of its output/input channels according to the production/consumption rate. For the sake of simplicity in the rest of this paper, we denote $[c_i]$ the number of initial tokens of the channel $c_i$.

\subsection{The PolyGraph Model}

\subsubsection{Syntax and semantics}

The PolyGraph model is a superset of the SDF model. It expands it in two ways: the production and consumption rates become periodic sequences (as in CSDF~\cite{bilsen_cyclostatic_1996}), and a subset of actors may have timing constraints. The terminology of the static analysis of an SDF specification, i.e., consistency, liveness, consistent and live execution, and hyperperiod~\cite{lee_synchronous_1987}, is extended to a PolyGraph specification. Consistency ensures the system can be executed within a bounded memory, while liveness guarantees deadlock-free execution. A hyperperiod is a partially ordered set of actors' jobs that returns the system to its initial state within a given time frame.

\subsubsection{Rational production and consumption rates}

Let $G = (V, E)$ be a PolyGraph specification and let $G_\Gamma = (\gamma_{ij}) \in \mathbb{Q}^{|E| \times |V|}$ be its topology matrix. The production and consumption rates are rational\footnote{A channel of a PolyGraph specification as defined in~\cite{dubrulle_polygraph_2021} has at least one integer rate. However, for the sake of simplicity, we consider in this paper that both rates are rational.}: a channel $c_i \in E$ is a tuple $(v_j, v_k, \gamma_{ij}, \gamma_{ik}, [c_i])$ such that $\gamma_{ij} \in \mathbb{Q}^*$ and $\gamma_{ik} \in \mathbb{Q}^*$, and $[c_i] \in \mathbb{Q}$ ($\mathbb{Q}^* = \mathbb{Q} \setminus \{ 0 \}$). Rational rates imply that the number of tokens produced/consumed can differ for each job. Specifically, a token is produced and stored in a channel if and only if sufficient fractional token parts are symbolically produced. Indeed, only an integer number of tokens may be produced or consumed. In other words, an actor with a rational production rate such as $1/n$ in the dataflow specification actually produces data once every $n$ jobs during runtime, thus validating the precedence constraint toward the consumer once every $n$ of its jobs.

The rational production and consumption rates are natural consequences of the periodicity of real-time actors (with specified and imposed frequencies) together with the consistency of the system's model imposed by the topology matrix. The rational rates are a compact notation for a CSDF equivalency~\cite{bilsen_cyclostatic_1996}.

\subsubsection{Timing constraints}

The actors of a PolyGraph specification may have a frequency constraint. Furthermore, they may also have a phase if they have a frequency constraint. The phase usually models the system's end-to-end latency. The frequency dictates the actor's execution at the specified rate, while the phase postpones its initial execution. Actors with frequency constraints are referred to as \emph{timed actors}. An execution of a consistent and live PolyGraph specification is an infinite repetition of its first hyperperiod. Although the number of jobs is not bounded during an execution, their timing constraints are cyclic over a hyperperiod.

\subsection{Disjoint Directed Meshes}

Let $G = (V, E)$ be a directed graph. Let us define two maps $in, out : V \to 2^E$ that associate to each actor $v \in V$ the set\footnote{$2^E$ is the powerset of $E$, that the set of all subset of $E$. For instance, if $E = \{ c_1, c_2 \}$, then $2^E = \{ \{ \emptyset \}, \{ c_1 \}, \{ c_2 \}, \{ c_1, c_2 \} \}$} of input and output channels, respectively. $G$ is a \emph{disjoint directed mesh} if:

\begin{enumerate}
  \item there is a single source actor, i.e., there exists a unique actor $v_s \in V$ such that $in(v_s) = \emptyset$;
  \item there is a single sink actor, i.e., there exists a unique actor $v_t \in V$ such that $out(v_t) = \emptyset$;
  \item there exists at least one path from $v_s$ to $v_t$, and a path is called a \emph{branch};
  \item the intersection of all the branches is the empty set.
\end{enumerate}

A disjoint directed mesh is illustrated in \cref{fig:disjoint-directed-mesh}.

\begin{figure}[htbp]
  \centering
  \resizebox{\columnwidth}{!}{
    \begin{tikzpicture}
      \node[actor] (src) at (0,0) {$source$};
      \node[actor] (a1) at (2,1) {$A_1$};
      \node[actor] (b1) at (2,0) {$B_1$};
      \node[actor] (c1) at (2,-1) {$C_1$};
      \node (adots) at (4,1) {$\dots$};
      \node (bdots) at (4,0) {$\dots$};
      \node (cdots) at (4,-1) {$\dots$};
      \node[actor] (a3) at (6,1) {$A_{n_1}$};
      \node[actor] (b3) at (6,0) {$B_{n_2}$};
      \node[actor] (c3) at (6,-1) {$C_{n_3}$};
      \node[actor] (sink) at (8,0) {$sink$};
      \draw[arc] (src) to (a1);
      \draw[arc] (src) to (b1);
      \draw[arc] (src) to (c1);
      \draw[arc] (a1) to (adots);
      \draw[arc] (b1) to (bdots);
      \draw[arc] (c1) to (cdots);
      \draw[arc] (adots) to (a3);
      \draw[arc] (bdots) to (b3);
      \draw[arc] (cdots) to (c3);
      \draw[arc] (a3) to (sink);
      \draw[arc] (b3) to (sink);
      \draw[arc] (c3) to (sink);
    \end{tikzpicture}
  }
  \caption{An illustration of a disjoint directed mesh with three branches. $n_1$ to $n_3$ are the (bounded positive) number of actors in each branch.}
  \label{fig:disjoint-directed-mesh}
\end{figure}
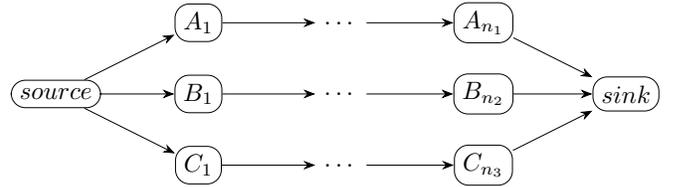

\section{The Routed PolyGraph Model: Syntax and Semantics}

\label{sec:routed-polygraph}

The RMDF model permits the specification of CPSs with conditional execution branches. This feature requires routing capabilities -specific actors to explicitly define how data tokens are routed into the conditional execution branches. We will introduce further in the paper such actors, called \emph{controlled splitters} and \emph{controlled joiners}. Controlled splitters and controlled joiners rely on actors called \emph{splitters}, \emph{joiners}, \emph{duplicaters}, and \emph{discard}~\cite{deoliveiracastro_reducing_2010,louise_graph_2019} . Those actors are not yet integrated into the PolyGraph model (only an implicit reference to them is used in~\cite{dubrulle_data_2019}). This section solves this issue by integrating these actors into the PolyGraph model, resulting in the \emph{routed PolyGraph model}.

\subsection{Routing Actors: Joiners, Splitter, Duplicaters, and Discards}

The actors \emph{Splitter}, \emph{Joiner}, \emph{Duplicater} and \emph{Discard} were introduced in~\cite{deoliveiracastro_reducing_2010}. These actors are routing actors. Instead of engaging in computational tasks, they primarily render data communication patterns explicit. This explicitness allows for compile-time optimizations~\cite{deoliveiracastro_reducing_2010}, especially reducing the memory requirements and improving overall CPS performance.

A \emph{splitter} consumes tokens from a single input and distributes them to multiple outputs. A predefined number of tokens is sent to the first channel in the lexicographic order of the output channels, then to the second channel, and so on. The number of tokens sent to an output channel is the numerator of the production rate of that channel. We enforce that an execution of a splitter consumes one token at a time; the reason behind this enforcement will be detailed in the next paragraph. Hence, the splitter of \cref{fig:splitter} sends two tokens to $c_1$ in two jobs and one token to $c_2$ in one job. The process is repeated cyclically. Similarly, a \emph{joiner} consumes tokens from multiple inputs and produces them to a single output. As for the splitter, the execution of a joiner produces one token at a time. Thus, the joiner of \cref{fig:joiner} receives two tokens from $c_1$ in two jobs and one token from $c_2$ in one job. The process is repeated cyclically. A \emph{duplicater}, as illustrated in \cref{fig:duplicater}, consumes tokens from a single input and sends them to multiple outputs. A \emph{discard} (cf. \cref{fig:discard}) is essentially a sink actor that consumes tokens from a single input without producing any. It is used to discard tokens that are no longer needed and is typically employed at the output of a splitter.

\begin{figure}[htbp]
  \centering
  \begin{subfigure}{\columnwidth}
    \centering
    \caption{Splitter.}
    \label{fig:splitter}
    \resizebox{\columnwidth}{!}{
      \begin{tikzpicture}
        \node[virtualactor] (spl) at (0,0) {Splitter};
        \draw[arc] (-2,0) -- node[pos=0,left] {$( 1, 2, 3, 4, 5, 6, \dots )$} (spl);
        \draw[arc] (spl) -- node[pos=0.2,sloped,above] {$2/3$} node[pos=0.5,sloped,below] {$c_1$} node[pos=1,right] {$( 1, 2, 4, 5, \dots )$} (3,0.8);
        \draw[arc] (spl) -- node[pos=0.2,sloped,below] {$1/3$} node[pos=0.5,sloped,above] {$c_2$} node[pos=1,right] {$( 3, 6, \dots )$} (3,-0.8);
      \end{tikzpicture}
    }
  \end{subfigure}
  \begin{subfigure}{\columnwidth}
    \centering
    \caption{Joiner.}
    \label{fig:joiner}
    \resizebox{\columnwidth}{!}{
      \begin{tikzpicture}
        \node[virtualactor] (joi) at (0,0) {Joiner};
        \draw[arc] (joi) -- node[pos=1,right] {$( 1, 2, 3, 4, 5, 6, \dots )$} (2,0);
        \draw[arc] (-3,0.8) -- node[pos=0,left] {$( 1, 2, 4, 5, \dots )$} node[pos=0.5,sloped,below] {$c_1$} node[pos=0.8,sloped,above] {$2/3$} (joi);
        \draw[arc] (-3,-0.8) -- node[pos=0,left] {$( 3, 6, \dots )$} node[pos=0.5,sloped,above] {$c_2$} node[pos=0.8,sloped,below] {$1/3$}  (joi);
      \end{tikzpicture}
    }
  \end{subfigure}
  \begin{subfigure}{\columnwidth}
    \centering
    \caption{Duplicater.}
    \label{fig:duplicater}
    \resizebox{\columnwidth}{!}{
      \begin{tikzpicture}
        \node[virtualactor] (dup) at (0,0) {Duplicater};
        \draw[arc] (-2,0) -- node[pos=0,left] {$( 1, 2, 3, \dots )$} (dup);
        \draw[arc] (dup) -- node[pos=0.5,sloped,below] {$c_1$} node[pos=1,right] {$( 1, 2, 3, \dots )$} (3,0.8);
        \draw[arc] (dup) -- node[pos=0.5,sloped,above] {$c_2$} node[pos=1,right] {$( 1, 2, 3, \dots )$} (3,-0.8);
      \end{tikzpicture}
    }
  \end{subfigure}
  \begin{subfigure}{\columnwidth}
    \centering
    \caption{Discard.}
    \label{fig:discard}
    \begin{tikzpicture}
      \node[virtualactor] (dis) at (0,0) {Discard};
      \draw[arc] (-2,0) -- node[pos=0,left] {$( 1, 2, 3, \dots )$} (dis);
    \end{tikzpicture}
  \end{subfigure}
  \caption{Illustration of a splitter, a joiner, a duplicater, and a discard with their input and output traces. Rates of $1$ and initial tokens of $0$ are omitted for clarity.}
\end{figure}
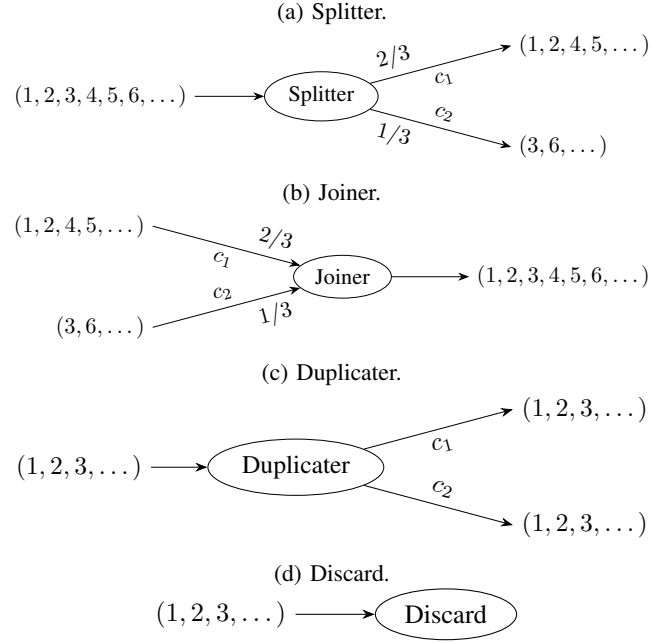

In order to ease the extension of static analysis from PolyGraph to the routed PolyGraph and then to RMDF, we enforce that an execution of a splitter or a joiner consumes or produces one token at a time. To that end, the production rate of the input/output channel of a splitter/joiner is equal to 1. In addition, the sum of the production/consumption rates of the output/input channels of a splitter/joiner is equal to 1. For instance, the production rates of the splitter of \cref{fig:splitter} are $2/3$ and $1/3$, whose sum equals 1.

\subsection{Static Analysis of a Routed PolyGraph Specification}

There is an equivalence between a PolyGraph specification and a routed PolyGraph specification. \cref{algo:removing-routing-actors} details the transformation of a routed PolyGraph specification to a semantically equivalent PolyGraph specification.

\begin{definition}[Floor and ceil functions]
  The floor function $\lfloor x \rfloor$ returns the greatest integer less than or equal to $x$. The ceil function $\lceil x \rceil$ returns the smallest integer greater than or equal to $x$.
\end{definition}

To illustrate, let $G = (V, E)$ be a routed PolyGraph specification with a splitter $v \in V$ that has $n$ output channels $c_1, \dots, c_n$ with the numerator of their production rate denoted as $\gamma_1, \dots, \gamma_n$. The splitter $v$ is removed by modifying the producer actor of the channel $c_i$ from $v$ to the predecessor of $v$. At this point, the predecessor of $v$ may send more than one token at some executions and no tokens at others. In order to reproduce the expected behavior of the splitter, one token has to be sent each time the predecessor of $v$ executes. This is done by shifting the production rate with initial tokens; a challenge here is to find the correct amount of initial tokens.

There is a bijection from the production rate and the initial tokens of a channel to the sequence of tokens sent to this channel. This bijection can be constructed by exhaustively enumerating all sequences generated by a production rate and all values of initial tokens (cf. \cref{algo:rate-and-initial-tokens-production-sequence}). For instance, the production sequence $[1, 1, 0]$ is associated with the production rate $2/3$. The questions is \enquote{how many initial tokens are needed to generate this production sequence?}. The production sequence $[1, 1, 0]$ is associated with the production rate $2/3$, so possible initial tokens are $[2/3, 1/3, 0]$. The production rate $2/3$ with no initial token generates the production sequence $[0, 1, 1]$, the production rate $2/3$ with $1/3$ initial token generates the production sequence $[1, 0, 1]$, and the production rate $2/3$ with $2/3$ initial token generates the production sequence $[1, 1, 0]$. We can see that the production sequence $[1, 1, 0]$ is associated with $2/3$ initial tokens.

The production sequence can be easily computed (cf. \cref{line:cons-sequence-1} to \cref{line:cons-sequence-2} of \cref{algo:removing-routing-actors}). By following the procedure described in \cref{algo:rate-and-initial-tokens-production-sequence}, the production rate and the required initial tokens can be derived. Initial tokens are associated to all output channels of a splitter with this approach bedore removing the splitter. As a consequence, the behavior of the splitter is preserved while removing it.

The same logic applies to removing a joiner using \cref{algo:rate-and-initial-tokens-consumption-sequence}. A discard is removed while removing the splittr, as a discard is always an output of a splitter. Removing a duplicater is trivial. \cref{fig:joiner-splitter-duplicater-transformation} illustrates the transformation of a routed PolyGraph specification to a semantically equivalent PolyGraph specification by removing a splitter (cf. \cref{fig:splitter-removal}), a joiner (cf. \cref{fig:joiner-removal}), a duplicater (cf. \cref{fig:duplicater-removal}), and a discard (cf. \cref{fig:discard-removal}).

\begin{algorithm}[htbp]
  \caption{Compute the rate and the number of initial tokens that generate a production sequence.}
  \label{algo:rate-and-initial-tokens-production-sequence}
  \KwIn{A production sequence $s$ of length $n$.}
  \KwOut{The production rate and the number of initial tokens generating the input production sequence.}
  $prod\_rate \gets \sum_{i = 1}^{n} s.get(i) / n$ \;
  $init\_tkn \gets 0$ \;
  \While{$init\_tkn \neq 1$}
  {
    $prod\_seq \gets []$ \;
    \For{$j \gets 1$ \KwTo $n$}
    {
      $prod\_seq.get(i) \gets \lfloor n \cdot \gamma + init\_tkn \rfloor - \lfloor (n - 1) \cdot \gamma + init\_tkn \rfloor$ \;
    }
    \If{$s = seq$}
    {
      \KwResult{$prod\_rate, init\_tkn$}
    }
    $init\_tkn \gets init\_tkn + 1/n$ \;
  }
\end{algorithm}

\begin{algorithm}[htbp]
  \caption{Compute the rate and the number of initial tokens that generate a consumption sequence.}
  \label{algo:rate-and-initial-tokens-consumption-sequence}
  \KwIn{A consumption sequence $s$ of length $n$.}
  \KwOut{The consumption rate and the number of initial tokens generating the input consumption sequence.}
  $cons\_rate \gets \sum_{i = 1}^{n} s.get(i) / n$ \;
  $init\_tkn \gets 0$ \;
  \While{$init\_tkn \neq 1$}
  {
    $cons\_seq \gets []$ \;
    \For{$j \gets 1$ \KwTo $n$}
    {
      $cons\_seq.get(i) \gets \lceil n \cdot \gamma - init\_tkn \rceil - \lceil (n - 1) \cdot \gamma - init\_tkn \rceil$ \;
    }
    \If{$s = cons\_seq$}
    {
      \KwResult{$cons\_rate, init\_tkn$}
    }
    $init\_tkn \gets init\_tkn + 1/n$ \;
  }
\end{algorithm}

\begin{algorithm}[htbp]
  \caption{Removal of routing actors from a routed Polygraph specification.}
  \label{algo:removing-routing-actors}
  \KwIn{A routed Polygraph specification $G = (V, E)$.}
  \KwOut{A semantically equivalent Polygraph specification without routing actors.}
  \For{$v \in V$}
  {
    \If{$v.is\_splitter$}
    {
      \For{$c \in v.output\_channels$}
      {
        $cons\_seq \gets []$ \;
        \For{$c' \in v.output\_channels$}
        {
          \label{line:cons-sequence-1}
          \eIf{$c'.index < c.index~\mathbf{or}~c'.index > c.index$}
          {
            $cons\_seq.append(0)$ \;
          }
          {
            $cons\_seq.append(1)$ \;
            \label{line:cons-sequence-2}
          }
        }
        $\_, init\_tkn \gets \cref{algo:rate-and-initial-tokens-consumption-sequence}~(cons\_seq, len(cons\_seq))$~;
        $E = E \cup \{ (v.pred, c.cons, c.prod\_rate, 1, [init\_tkn]) \} $\;
      }
    }
    \If{$v.is\_joiner$}
    {
      \For{$c \in v.input\_channels$}
      {
        $prod\_seq \gets []$ \;
        \For{$c' \in v.input\_channels$}
        {
          \eIf{$c'.index < c.index~\mathbf{or}~c'.index > c.index$}
          {
            $prod\_seq.append(0)$ \;
          }
          {
            $prod\_seq.append(1)$ \;
          }
        }
        $\_, init\_tkn \gets \cref{algo:rate-and-initial-tokens-production-sequence}~ (prod\_seq, len(prod\_seq))$~;
        $E = E \cup \{ (c.prod, v.succ, 1, c.cons\_rate, [init\_tkn]) \}$\;
      }
    }
    \If{$v.is\_duplicater$}
    {
      \For{$c \in v.output\_channels$}
      {
        $E = E \cup \{ (v.pred, c.cons, 1, 1, [0]) \}$\;
      }
    }
  }
  \For{$v \in V$}
  {
    \If{$v.is\_splitter$ \textbf{or} $v.is\_duplicater$}
    {
      \For{$c \in v.output\_channels$}
      {
        $E = E \setminus \{ c \}$\;
      }
      $V = V \setminus \{ v \}$\;
    }
    \If{$v.is\_joiner$}
    {
      \For{$c \in v.input\_channels$}
      {
        $E = E \setminus \{ c \}$\;
      }
      $V = V \setminus \{ v \}$\;
    }
    \If{$v.is\_discard$}
    {
      \For{$c \in v.input\_channels$}
      {
        $E = E \setminus \{ c \}$\;
      }
      $V = V \setminus \{ v \}$\;
    }
  }
  \KwResult{$G = (V, E)$}
\end{algorithm}

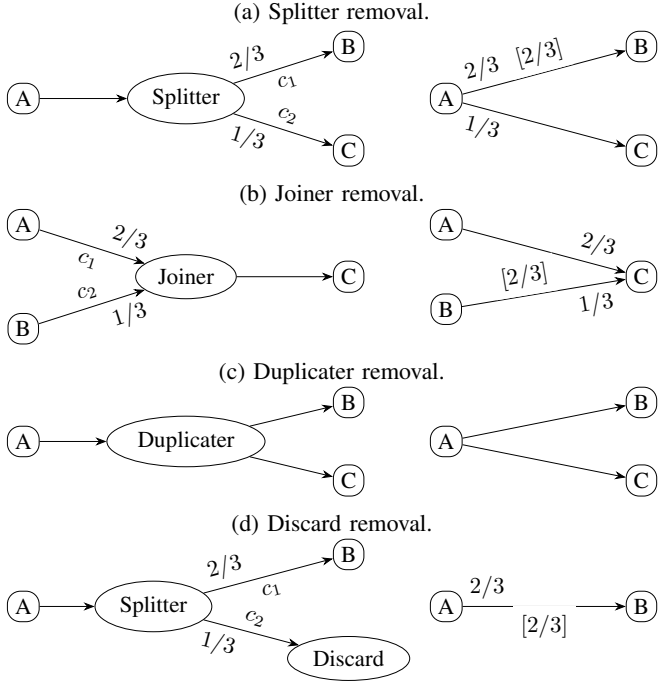
\begin{figure}[htbp]
  \centering
  \begin{subfigure}[c]{\columnwidth}
    \centering
    \caption[Splitter removal]{Splitter removal.}
    \label{fig:splitter-removal}
    \resizebox{\columnwidth}{!}{
      \begin{tikzpicture}
        \node[virtualactor] (spl) at (0,2) {Splitter};
        \node[actor] (a1) at (-2.5,2) {A};
        \node[actor] (b1) at (2.5,2.8) {B};
        \node[actor] (c1) at (2.5,1.2) {C};
        \draw[arc] (a1) -- (spl);
        \draw[arc] (spl) -- node[pos=0.2,sloped,above] {$2/3$} node[pos=0.5,sloped,below] {$c_1$} (b1);
        \draw[arc] (spl) -- node[pos=0.2,sloped,below] {$1/3$} node[pos=0.5,sloped,above] {$c_2$} (c1);
        \node[actor] (a12) at (4,2) {A};
        \node[actor] (b12) at (7, 2.8) {B};
        \node[actor] (c12) at (7,1.2) {C};
        \draw[arc] (a12) -- node[pos=0.15,sloped,above] {$2/3$} node[pos=0.5,sloped,fill=white,above] {$[2/3]$} (b12);
        \draw[arc] (a12) -- node[pos=0.15,sloped,below] {$1/3$} (c12);
      \end{tikzpicture}
    }
  \end{subfigure}
  \begin{subfigure}[c]{\columnwidth}
    \centering
    \caption[Joiner removal]{Joiner removal.}
    \label{fig:joiner-removal}
    \resizebox{\columnwidth}{!}{
      \begin{tikzpicture}
        \node[virtualactor] (joi) at (0,-0.6) {Joiner};
        \node[actor] (a2) at (2.5,-0.6) {C};
        \node[actor] (b2) at (-2.5,0.2) {A};
        \node[actor] (c2) at (-2.5,-1.4) {B};
        \draw[arc] (joi) -- (a2);
        \draw[arc] (b2) -- node[pos=0.5,sloped,below] {$c_1$} node[pos=0.8,sloped,above] {$2/3$} (joi);
        \draw[arc] (c2) -- node[pos=0.5,sloped,above] {$c_2$} node[pos=0.8,sloped,below] {$1/3$} (joi);
        \node[actor] (a22) at (7,-0.6) {C};
        \node[actor] (b22) at (4,0.2) {A};
        \node[actor] (c22) at (4,-1.1) {B};
        \draw[arc] (b22) -- node[pos=0.8,sloped,above] {$2/3$} (a22);
        \draw[arc] (c22) -- node[pos=0.8,sloped,below] {$1/3$} node[pos=0.4,sloped,fill=white,above] {$[2/3]$} (a22);
      \end{tikzpicture}
    }
  \end{subfigure}
  \begin{subfigure}[c]{\columnwidth}
    \centering
    \caption[Duplicater removal]{Duplicater removal.}
    \label{fig:duplicater-removal}
    \resizebox{\columnwidth}{!}{
      \begin{tikzpicture}
        \node[virtualactor] (dup) at (0,-3.2) {Duplicater};
        \node[actor] (a3) at (-2.5,-3.2) {A};
        \node[actor] (b3) at (2.5,-2.6 ) {B};
        \node[actor] (c3) at (2.5,-3.8) {C};
        \draw[arc] (a3) -- (dup);
        \draw[arc] (dup) -- (b3);
        \draw[arc] (dup) -- (c3);
        \node[actor] (a32) at (4,-3.2) {A};
        \node[actor] (b32) at (7,-2.6) {B};
        \node[actor] (c32) at (7,-3.8) {C};
        \draw[arc] (a32) -- (b32);
        \draw[arc] (a32) -- (c32);
      \end{tikzpicture}
    }
  \end{subfigure}
  \begin{subfigure}[c]{\columnwidth}
    \centering
    \caption[Discard removal]{Discard removal.}
    \label{fig:discard-removal}
    \resizebox{\columnwidth}{!}{
      \begin{tikzpicture}
        \node[virtualactor] (spl) at (-0.5,-5.6) {Splitter};
        \node[actor] (a1) at (-2.5,-5.6) {A};
        \node[actor] (b1) at (2.5,-4.8) {B};
        \node[virtualactor] (c1) at (2.5,-6.4) {Discard};
        \draw[arc] (a1) -- (spl);
        \draw[arc] (spl) -- node[pos=0.2,sloped,above] {$2/3$} node[pos=0.5,sloped,below] {$c_1$} (b1);
        \draw[arc] (spl) -- node[pos=0.2,sloped,below] {$1/3$} node[pos=0.5,sloped,above] {$c_2$} (c1);
        \node[actor] (a12) at (4,-5.6) {A};
        \node[actor] (b12) at (7, -5.6) {B};
        \draw[arc] (a12) -- node[pos=0.15,sloped,above] {$2/3$} node[pos=0.5,sloped,fill=white,below] {$[2/3]$} (b12);
      \end{tikzpicture}
    }
  \end{subfigure}
  \caption{Removal of routing actors while preserving the semantics of data communication pattern. Rates of $1$ and initial tokens of $0$ are omitted for clarity.}
  \label{fig:joiner-splitter-duplicater-transformation}
\end{figure}

The resulting PolyGraph specification is less explicit as there are no routing actors. However, the static analyses of PolyGraph apply to the transformed PolyGraph specification.

\section{RMDF model: Syntax and Semantics}

\label{sec:rmdf}

This section explains how splitters and joiners are extended to the controlled splitters and controlled joiners. Those two latter actors will be used to route data tokens into conditional execution branches in an RMDF specification.

\subsection{Controlled Routing Actors and Mode-Dependent Execution}

A controlled splitter/joiner is a splitter/joiner with an additional input \emph{control channel}. A \emph{control channel} is where transit \emph{control token}. A controlled splitter/joiner, which is illustrated in \cref{fig:controlled-splitter} and \cref{fig:controlled-joiner}, alternates between consuming a control token and consuming/producing a data token. The control token decides which channel the next data token will be sent/received, leading to a \emph{mode-dependent} execution. In addition, the production/consumption rates of the output/input channels of a controlled splitter/joiner are parameters that take the value 0 or 1. For instance, parameters\footnote{$m_1$ and $m_2$ are symbolic rates which are assigned a value at runtime} $m_1$ and $m_2$ in \cref{fig:controlled-splitter} and \cref{fig:controlled-joiner} take the value 0 or 1. The actual value of $m_1$ and $m_2$ is determined at runtime by the value carried by control tokens consumed from $c_2$ in \cref{fig:controlled-splitter} and \cref{fig:controlled-joiner}.

As for the non-controlled splitter/joiner, we want to ensure that executing a controlled splitter/joiner consumes/produces one token at a time. To that end, we enforce that a parameter assignation (i.e., transforming $m_1$ and $m_2$ into 0 or 1 at runtime) to adhere to the following constraints: exactly one parameter must be equal to 1, and the other parameters must be equal to 0. In \cref{fig:controlled-splitter-joiner}, we have either $m_1 = 1$ and $m_2$ = 0 or $m_1 = 0$ and $m_2 = 1$. We will detail later in the paper how this constraint is formally defined.

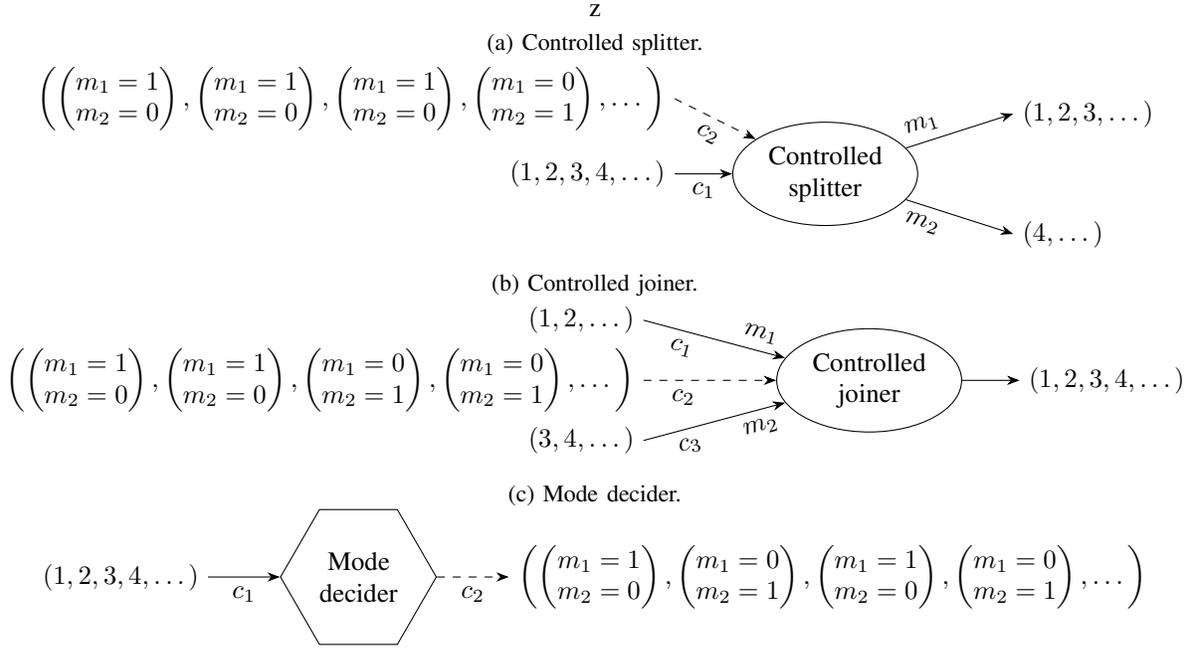
\begin{figure*}[htbp]
  \centering
  \begin{subfigure}{\textwidth}z
    \centering
    \caption{Controlled splitter.}
    \label{fig:controlled-splitter}
    \begin{tikzpicture}
      \node[virtualactor] (spl) at (0,0) {Controlled \\ splitter};
      \draw[arc] (-2,0) -- node[pos=0,left] {$( 1, 2, 3, 4, \dots )$} node[pos=0.5,sloped,below] {$c_1$} (spl);
      \draw[arccontrol] (-2,1) -- node[pos=0,left] {$\biggl( \begin{pmatrix} m_1 = 1 \\ m_2 = 0 \end{pmatrix}, \begin{pmatrix} m_1 = 1 \\ m_2 = 0 \end{pmatrix}, \begin{pmatrix} m_1 = 1 \\ m_2 = 0 \end{pmatrix}, \begin{pmatrix} m_1 = 0 \\ m_2 = 1 \end{pmatrix}, \dots \biggl)$} node[pos=0.5,sloped,below] {$c_2$} (spl);
      \draw[arc] (spl) -- node[pos=0.2,sloped,above] {$m_1$} node[pos=1,right] {$( 1, 2, 3, \dots )$} (2.5,0.8);
      \draw[arc] (spl) -- node[pos=0.2,sloped,below] {$m_2$} node[pos=1,right] {$( 4, \dots )$} (2.5,-0.8);
    \end{tikzpicture}
  \end{subfigure}
  \begin{subfigure}{\textwidth}
    \centering
    \caption{Controlled joiner.}
    \label{fig:controlled-joiner}
    \begin{tikzpicture}
      \node[virtualactor] (joi) at (0,0) {Controlled \\ joiner};
      \draw[arc] (joi) -- node[pos=1,right] {$( 1, 2, 3, 4, \dots )$} (2,0);
      \draw[arccontrol] (-3,0) -- node[pos=0,left] {$\biggl( \begin{pmatrix} m_1 = 1 \\ m_2 = 0 \end{pmatrix}, \begin{pmatrix} m_1 = 1 \\ m_2 = 0 \end{pmatrix}, \begin{pmatrix} m_1 = 0 \\ m_2 = 1 \end{pmatrix}, \begin{pmatrix} m_1 = 0 \\ m_2 = 1 \end{pmatrix}, \dots \biggl)$} node[pos=0.3,sloped,below] {$c_2$} (joi);
      \draw[arc] (-3,0.8) -- node[pos=0,left] {$( 1, 2, \dots )$} node[pos=0.3,sloped,below] {$c_1$} node[pos=0.8,sloped,above] {$m_1$}  (joi);
      \draw[arc] (-3,-0.8) -- node[pos=0,left] {$( 3, 4, \dots )$} node[pos=0.3,sloped,below] {$c_3$} node[pos=0.8,sloped,below] {$m_2$}  (joi);
    \end{tikzpicture}
  \end{subfigure}
  \begin{subfigure}{\textwidth}
    \centering
    \caption{Mode decider.}
    \label{fig:mode-decider}
    \begin{tikzpicture}
      \node[decider] (md) at (0,0) {Mode \\ decider};
      \draw[arc] (-2,0) -- node[midway,below] {$c_1$} node[pos=0,left] {$( 1, 2, 3, 4, \dots )$} (md);
      \draw[arc,dashed] (md) -- node[midway,below] {$c_2$} node[pos=1,right] {$\biggl( \begin{pmatrix} m_1 = 1 \\ m_2 = 0 \end{pmatrix}, \begin{pmatrix} m_1 = 0 \\ m_2 = 1 \end{pmatrix}, \begin{pmatrix} m_1 = 1 \\ m_2 = 0 \end{pmatrix}, \begin{pmatrix} m_1 = 0 \\ m_2 = 1 \end{pmatrix}, \dots \biggl)$} (2,0);
    \end{tikzpicture}
  \end{subfigure}
  \caption{An illustration of a controlled splitter and a controlled joiner with their input and output traces. Dashed lines represent control channels, and plain lines represent data channels. Rates of $1$ and initial tokens of $0$ are omitted for clarity.}
  \label{fig:controlled-splitter-joiner}
\end{figure*}

\subsection{Mode Decider}

As controlled splitters and joiners consume control tokens, we need an actor to produce them. A \emph{mode decider} consumes a data token from its input channels and produces a control token on its output control channel. The algorithm inside the mode decider that determines the value of the control token is usually a switch case, but it can be more complex depending on the application's needs. A mode decider has exactly one input data channel\footnote{It would not be difficult to consider more than one input data channel for mode deciders.} (the channel $c_1$ in \cref{fig:mode-decider}) and exactly one output control channel (the channel $c_2$ in \cref{fig:mode-decider}), with this latter being connected to a duplicater. This duplicater is connected to one or many controlled splitters and controlled joiners. \cref{fig:mode-decider} illustrates a mode decider.

\subsection{A Complete RMDF Specification}

\cref{fig:example-rmdf} illustrates an RMDF specification with three conditional execution branches (C$_1$ to C$_{n_1}$, D$_1$ to D$_{n_2}$ and E$_1$ to E$_{n_3}$ with $n_1, n_2, n_3 \in \mathbb{N}^*$). An RMDF specification is a tuple $(V, E, M)$ where $V$ is a finite set of actors, $E$ is a finite set of channels, and $M$ is a set of functions that enforce the production and consumption of one token at a time for the controlled splitters and joiners. This set of functions of the RMDF specification of \cref{fig:example-rmdf} is $M = \{ M_1, M_2, M_3 \}$ as defined in the right bottom corner of \cref{fig:example-rmdf}. Each function specifies an execution mode. In \cref{fig:example-rmdf}, the actor B is a mode decider. The control token it produces is sent to a duplicater and then to a controlled splitter and a controlled joiner. The value carried by the control token (produced by B) will be used at runtime to determine the value of parameters $m_1$, $m_2$ and $m_3$.

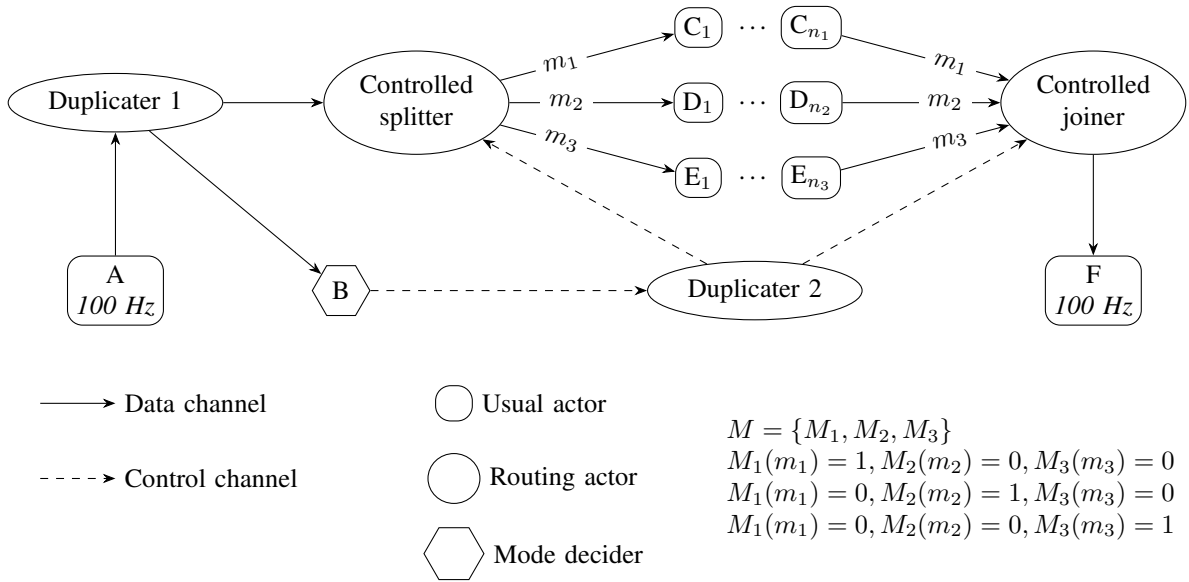
\begin{figure*}[htbp]
  \centering
  \begin{tikzpicture}
    \node[actor] (a) at (-1,-2.5) {A \\ \emph{100 Hz}};
    \node[virtualactor] (dup-1) at (-1,0) {Duplicater 1};
    \node[virtualactor] (dup-2) at (7.5,-2.5) {Duplicater 2};
    \node[virtualactor] (cs) at (3,0) {Controlled \\ splitter};
    \node[virtualactor] (cj) at (12,0) {Controlled \\ joiner};
    \node[decider] (b) at (2,-2.5) {B};
    \node[actor] (c1) at (6.75,1) {C$_1$};
    \node (ci) at (7.5,1) {\dots};
    \node[actor] (cl) at (8.25,1) {C$_{n_1}$};
    \node[actor] (d1) at (6.75,0) {D$_1$};
    \node (di) at (7.5,0) {\dots};
    \node[actor] (dm) at (8.25,0) {D$_{n_2}$};
    \node[actor] (e1) at (6.75,-1) {E$_1$};
    \node (ei) at (7.5,-1) {\dots};
    \node[actor] (en) at (8.25,-1) {E$_{n_3}$};
    \node[actor] (f) at (12,-2.5) {F \\ \emph{100 Hz}};
    \draw[arc] (a) to (dup-1);
    \draw[arc] (cj) to (f);
    \draw[arc] (dup-1) to (cs);
    \draw[arc] (dup-1) to (b);
    \draw[arc] (cs) to node[sloped,fill=white,pos=0.35] {$m_1$} (c1);
    \draw[arc] (cs) to node[sloped,fill=white,pos=0.35] {$m_2$} (d1);
    \draw[arc] (cs) to node[sloped,fill=white,pos=0.35] {$m_3$} (e1);
    \draw[arc] (en) to node[sloped,fill=white,pos=0.65] {$m_3$} (cj);
    \draw[arc] (dm) to node[sloped,fill=white,pos=0.65] {$m_2$} (cj);
    \draw[arc] (cl) to node[sloped,fill=white,pos=0.65] {$m_1$} (cj);
    \draw[arc,dashed] (b) to (dup-2);
    \draw[arc,dashed] (dup-2) to (cs);
    \draw[arc,dashed] (dup-2) to (cj);
    \draw[arc] (-2,-4) to (-1,-4);
    \node[anchor=west] at (-1,-4) {Data channel};
    \draw[arc,dashed] (-2,-5) to (-1,-5);
    \node[anchor=west] at (-1,-5) {Control channel};
    \node[actor] (toto) at (3.5,-4) {\textcolor{white}{A}};
    \node[anchor=west] at (toto.east) {Usual actor};
    \node[decider, aspect=0.5] (titi) at (3.5,-6) {\textcolor{white}{A}};
    \node[anchor=west] at (titi.east) {Mode decider};
    \node[virtualactor] (tutu) at (3.5,-5) {\textcolor{white}{A}};
    \node[anchor=west] at (tutu.east) {Routing actor};
    \node[anchor=west,align=left] at (7,-5) {$M = \{ M_1, M_2, M_3 \} $ \\ $M_1(m_1) = 1, M_2(m_2) = 0, M_3(m_3) = 0$ \\ $M_1(m_1) = 0, M_2(m_2) = 1, M_3(m_3) = 0$ \\ $M_1(m_1) = 0, M_2(m_2) = 0, M_3(m_3) = 1$};
  \end{tikzpicture}
  \caption{An example of an RMDF specification. A usual actor is an actor which is neither (non-) controlled splitter or joiner, nor a mode decider.}
  \label{fig:example-rmdf}
\end{figure*}

\section{Static Analysis of an RMDF Specification}

\label{sec:static-analysis}

\subsection{Control Areas}

In order to ease the extension of the static analysis of the PolyGraph model to the RMDF model, we enforce that an execution of a controlled splitter/joiner consumes/produces one token at a time. In addition, we impose conditions on using the mode deciders, controlled splitters, and controlled joiners in an RMDF specification. Let us define \emph{control areas} (\cref{def:control-area}) on which restrictions (\cref{prop:restrictions-control-areas}) will be applied. A control area is a disjoint directed mesh where the source and sink actors are a controlled splitter and controlled joiner.

\begin{definition}[Control area of a mode decider]
  \label{def:control-area}
  Let $G = (V, E, M)$ be an RMDF specification, let $v_{md} \in V$ be a mode decider, let $dup(v_{md})$ be the duplicater connected to $v_{md}$, let $cs(v_{md})$ be the controlled splitters connected to $dup(v_{md})$ and let $cj(v_{md})$ be the controlled joiners connected to $dup(v_{md})$. The \emph{control area} of $v_{md}$, defined as $control(v_{md})$, is the set of actors conditioned by the execution mode output by $v_{md}$. It is the intersection of all the successors of controlled splitters belonging to $cs(v_{md})$ and all the predecessors of the controlled joiners belonging to $cj(v_{md})$. Formally, we have $control(v_{md}) = \{ v \in V | \exists (v_1, \dots, v_i, \dots, v_n), v_1 \in cs(v_{md}), v_i = v, v_n \in cj(v_{md}), \forall 1 \leq j < n : v_{j+1} \in succ(v_j) \}$ with $succ(v)$ the set of successors of $v$.
\end{definition}

For instance, in \cref{fig:example-rmdf}, the actor B is a mode decider. We have $dup(B) = \{ Duplicater~2$ \}, $cs(B) = \{ Controlled~splitter \}$ and $cj(B) = \{ Controlled~joiner \}$. The control area of the actor $B$ is $\{ C_1, \dots, C_{n_1}, D_1, \dots, D_{n_2}, E_1, \dots, E_{n_3} \}$. The following property imposes restrictions on the control areas of an RMDF specification.

\begin{property}[Restrictions on controls areas]
  \label{prop:restrictions-control-areas}
  Let $G = (V, E, M)$ be an RMDF specification:
  \begin{enumerate}
    \item Execution of actors in a control area must depend on exactly one execution mode. \cref{algo:mode-coherence} return \textbf{True} if this condition is fulfilled, \textbf{False} otherwise (see \cref{fig:control-area-1} for an RMDF specification that does not fulfill this requirement).
    \item Let $v_{md} \in V$ be a mode decider. There is no channel from an actor $v_1 \in control(v_{md})$ / $v_1 \notin control(v_{md})$ to an actor $v_2 \notin control(v_{md})$ / $v_2 \in control(v_{md})$ (cf. \cref{fig:control-area-2}).
    \item If there are timed actors in a control area, they must have the same frequency (cf. \cref{fig:control-area-3}).
    \item The production and consumption rate inside a control area and the parametric rates associated with the controlled joiner and controlled splitter of that control area must be equal to 1 or 0 (cf. \cref{fig:control-area-4}).
    \item Let $v_{md} \in V$ be a mode decider. When all parametric rates of a controlled splitter are assigned a value, their sum is equal to 1. The same condition holds for the parametric rates of a controlled joiner (cf. \cref{fig:control-area-5}).
  \end{enumerate}
\end{property}

\begin{algorithm}[htbp]
  \caption{Checking that actors' execution in control areas of an RMDF specification depends on a single execution mode. Some work remains to handle loops within a branch.}
  \label{algo:mode-coherence}
  \KwIn{An RMDF specification $G = (V, E, M)$.}
  \KwOut{$\mathbf{True}$ if execution of the actors of $G$ depends on a single mode, $\mathbf{False}$ otherwise.}
  \Def{label\_successor(v, associated\_mode)}
  {
    \For{$c \in v.output\_channels$}
    {
      \If{$mode \notin c.conditioning\_mode$}
      {
        $c.conditioning\_mode.append(mode)$ \;
      }
      \If{$len(c.conditioning\_mode) = 2$}
      {
        \KwResult{$\mathbf{False}$}
      }
      \eIf{$c.consumer.is\_controlled\_joiner$}
      {
        \KwResult{$\mathbf{True}$}
      }
      {
        $label\_successor(c.consumer, mode)$ \;
      }
    }
  }
  \For{$v \in V$}
  {
    \If{$v.is\_mode\_decider$}
    {
      \For{$v \in cs(v)$}
      {
        $mode = 0$ \;
        \For{$c \in v.output\_channels$}
        {
          \If{$label\_successors(c.consumer, mode) = \mathbf{False}$}
          {
            \KwResult{$\mathbf{False}$}
          }
          $mode \gets mode + 1$ \;
        }
      }
    }
  }
  \KwResult{$\mathbf{True}$}
\end{algorithm}

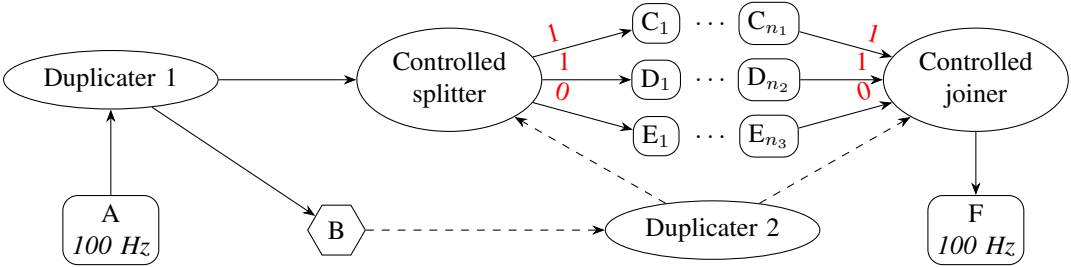
\begin{figure*}[htbp]
  \begin{subfigure}{\textwidth}
    \centering
    \caption{Actors C$_1$ and D$_{n_2}$ are not in two distincts branches.}
    \label{fig:control-area-1}
    \begin{tikzpicture}
      \node[actor] (a) at (-1,-2) {A \\ \emph{100 Hz}};
      \node[virtualactor] (dup-1) at (-1,0) {Duplicater 1};
      \node[virtualactor] (dup-2) at (7,-2) {Duplicater 2};
      \node[virtualactor] (cs) at (3.5,0) {Controlled \\ splitter};
      \node[virtualactor] (cj) at (10.5,0) {Controlled \\ joiner};
      \node[decider] (b) at (2,-2) {B};
      \node[actor] (c) at (6.25,0.75) {C$_1$};
      \node (dots1) at (7,0.75) {$\dots$};
      \node[actor] (c2) at (7.75,0.75) {C$_{n_1}$};
      \node[actor] (d) at (6.25,0) {D$_1$};
      \node (dots2) at (7,0) {$\dots$};
      \node[actor] (d2) at (7.75,0) {D$_{n_2}$};
      \node[actor] (e) at (6.25,-0.75) {E$_1$};
      \node (dots3) at (7,-0.75) {$\dots$};
      \node[actor] (e2) at (7.75,-0.75) {E$_{n_3}$};
      \node[actor] (f) at (10.5,-2) {F \\ \emph{100 Hz}};
      \draw[arc] (a) to (dup-1);
      \draw[arc] (cj) to (f);
      \draw[arc] (dup-1) to (b);
      \draw[arc] (dup-1) to (cs);
      \draw[arc] (cs) to node[sloped,fill=white,pos=0.35] {$m_1$} (c);
      \draw[arc] (cs) to node[sloped,fill=white,pos=0.35] {$m_1$} (d);
      \draw[arc] (cs) to node[sloped,fill=white,pos=0.35] {$m_3$} (e);
      \draw[arc] (c2) to node[sloped,fill=white,pos=0.6] {$m_1$} (cj);
      \draw[arc] (d2) to node[sloped,fill=white,pos=0.6] {$m_2$} (cj);
      \draw[arc] (e2) to node[sloped,fill=white,pos=0.6] {$m_3$} (cj);
      \draw[arc,line width=1.5pt,red] (c) to (d2);
      \draw[arc,dashed] (b) to (dup-2);
      \draw[arc,dashed] (dup-2) to (cj);
      \draw[arc,dashed] (dup-2) to (cs);
    \end{tikzpicture}
  \end{subfigure}
  \begin{subfigure}{\textwidth}
    \centering
    \caption{There is a channel from an actor outside the control area to an actor inside the control area, and conversely.}
    \label{fig:control-area-2}
    \begin{tikzpicture}
      \clip (-3.6,2) rectangle (12.4,-2.7);
      \node[actor] (a) at (-1,-2) {A \\ \emph{100 Hz}};
      \node[virtualactor] (dup-1) at (-1,0) {Duplicater 1};
      \node[virtualactor] (dup-2) at (7,-2) {Duplicater 2};
      \node[virtualactor] (cs) at (3.5,0) {Controlled \\ splitter};
      \node[virtualactor] (cj) at (10.5,0) {Controlled \\ joiner};
      \node[decider] (b) at (2,-2) {B};
      \node[actor] (c) at (6.25,0.75) {C$_1$};
      \node (dots1) at (7,0.75) {$\dots$};
      \node[actor] (c2) at (7.75,0.75) {C$_{n_1}$};
      \node[actor] (d) at (6.25,0) {D$_1$};
      \node (dots2) at (7,0) {$\dots$};
      \node[actor] (d2) at (7.75,0) {D$_{n_2}$};
      \node[actor] (e) at (6.25,-0.75) {E$_1$};
      \node (dots3) at (7,-0.75) {$\dots$};
      \node[actor] (e2) at (7.75,-0.75) {E$_{n_3}$};
      \node[actor] (f) at (10.5,-2) {F \\ \emph{100 Hz}};
      \draw[arc] (a) to (dup-1);
      \draw[arc] (cj) to (f);
      \draw[arc] (dup-1) to (b);
      \draw[arc] (dup-1) to (cs);
      \draw[arc] (cs) to node[sloped,fill=white,pos=0.35] {$m_1$} (c);
      \draw[arc] (cs) to node[sloped,fill=white,pos=0.35] {$m_1$} (d);
      \draw[arc] (cs) to node[sloped,fill=white,pos=0.35] {$m_3$} (e);
      \draw[arc] (c2) to node[sloped,fill=white,pos=0.6] {$m_1$} (cj);
      \draw[arc] (d2) to node[sloped,fill=white,pos=0.6] {$m_2$} (cj);
      \draw[arc] (e2) to node[sloped,fill=white,pos=0.6] {$m_3$} (cj);
      \draw[arc,line width=1.5pt,red] (a) to[out=140,in=170,looseness=2] (c);
      \draw[arc,line width=1.5pt,red] (c) to[out=20,in=40,looseness=2.2] (f);
      \draw[arc,dashed] (b) to (dup-2);
      \draw[arc,dashed] (dup-2) to (cj);
      \draw[arc,dashed] (dup-2) to (cs);
    \end{tikzpicture}
  \end{subfigure}
  \begin{subfigure}{\textwidth}
    \centering
    \caption{Actors C, D, and E have different frequencies. For readability, each conditional execution branch has exactly one actor.}
    \label{fig:control-area-3}
    \begin{tikzpicture}
      \node[actor] (a) at (-1,-2) {A \\ \emph{100 Hz}};
      \node[virtualactor] (dup-1) at (-1,0) {Duplicater 1};
      \node[virtualactor] (dup-2) at (7,-2) {Duplicater 2};
      \node[virtualactor] (cs) at (3.5,0) {Controlled \\ splitter};
      \node[virtualactor] (cj) at (10.5,0) {Controlled \\ joiner};
      \node[decider] (b) at (2,-2) {B};
      \node[actor,fill=red!20] (c) at (7,1) {C \\ \emph{25 Hz}};
      \node[actor,fill=red!20] (d) at (7,0) {D \\ \emph{50 Hz}};
      \node[actor,fill=red!20] (e) at (7,-1) {E \\ \emph{100 Hz}};
      \node[actor] (f) at (10.5,-2) {F \\ \emph{100 Hz}};
      \draw[arc] (a) to (dup-1);
      \draw[arc] (cj) to (f);
      \draw[arc] (dup-1) to (b);
      \draw[arc] (dup-1) to (cs);
      \draw[arc] (cs) to node[sloped,fill=white,pos=0.35] {$m_1$} (c);
      \draw[arc] (cs) to node[sloped,fill=white,pos=0.35] {$m_1$} (d);
      \draw[arc] (cs) to node[sloped,fill=white,pos=0.35] {$m_3$} (e);
      \draw[arc] (c) to node[sloped,fill=white,pos=0.6] {$m_1$} (cj);
      \draw[arc] (d) to node[sloped,fill=white,pos=0.6] {$m_2$} (cj);
      \draw[arc] (e) to node[sloped,fill=white,pos=0.6] {$m_3$} (cj);
      \draw[arc,dashed] (b) to (dup-2);
      \draw[arc,dashed] (dup-2) to (cj);
      \draw[arc,dashed] (dup-2) to (cs);
    \end{tikzpicture}
  \end{subfigure}
  \begin{subfigure}{\textwidth}
    \centering
    \caption{At least one production or consumption rate of the control area is not equal to 0 or 1.}
    \label{fig:control-area-4}
    \begin{tikzpicture}
      \node[actor] (a) at (-1,-2) {A \\ \emph{100 Hz}};
      \node[virtualactor] (dup-1) at (-1,0) {Duplicater 1};
      \node[virtualactor] (dup-2) at (7,-2) {Duplicater 2};
      \node[virtualactor] (cs) at (3.5,0) {Controlled \\ splitter};
      \node[virtualactor] (cj) at (10.5,0) {Controlled \\ joiner};
      \node[decider] (b) at (2,-2) {B};
      \node[actor] (c) at (6.25,0.75) {C$_1$};
      \node (dots1) at (7,0.75) {$\dots$};
      \node[actor] (c2) at (7.75,0.75) {C$_{n_1}$};
      \node[actor] (d) at (6.25,0) {D$_1$};
      \node (dots2) at (7,0) {$\dots$};
      \node[actor] (d2) at (7.75,0) {D$_{n_2}$};
      \node[actor] (e) at (6.25,-0.75) {E$_1$};
      \node (dots3) at (7,-0.75) {$\dots$};
      \node[actor] (e2) at (7.75,-0.75) {E$_{n_3}$};
      \node[actor] (f) at (10.5,-2) {F \\ \emph{100 Hz}};
      \draw[arc] (a) to (dup-1);
      \draw[arc] (cj) to (f);
      \draw[arc] (dup-1) to (b);
      \draw[arc] (dup-1) to (cs);
      \draw[arc] (cs) to node[sloped,fill=white,pos=0.35] {$m_1$} node[very near end, sloped, above] {\textcolor{red}{2}} (c);
      \draw[arc] (cs) to node[sloped,fill=white,pos=0.35] {$m_2$} (d);
      \draw[arc] (cs) to node[sloped,fill=white,pos=0.35] {$m_3$} (e);
      \draw[arc] (c2) to node[sloped,fill=white,pos=0.6] {$m_1$} (cj);
      \draw[arc] (d2) to node[sloped,fill=white,pos=0.6] {$m_2$} (cj);
      \draw[arc] (e2) to node[sloped,fill=white,pos=0.6] {$m_3$} (cj);
      \draw[arc,dashed] (b) to (dup-2);
      \draw[arc,dashed] (dup-2) to (cj);
      \draw[arc,dashed] (dup-2) to (cs);
    \end{tikzpicture}
  \end{subfigure}
  \begin{subfigure}{\textwidth}
    \centering
    \caption{The sum of the assigned values of the production rates/consumption rates of the controlled splitter/controlled joiner is not equal to 1. $m_1$ and $m_2$ are assigned to 1 and $m_3$ to 0.}
    \label{fig:control-area-5}
    \begin{tikzpicture}
      \node[actor] (a) at (-1,-2) {A \\ \emph{100 Hz}};
      \node[virtualactor] (dup-1) at (-1,0) {Duplicater 1};
      \node[virtualactor] (dup-2) at (7,-2) {Duplicater 2};
      \node[virtualactor] (cs) at (3.5,0) {Controlled \\ splitter};
      \node[virtualactor] (cj) at (10.5,0) {Controlled \\ joiner};
      \node[decider] (b) at (2,-2) {B};
      \node[actor] (c) at (6.25,0.75) {C$_1$};
      \node (dots1) at (7,0.75) {$\dots$};
      \node[actor] (c2) at (7.75,0.75) {C$_{n_1}$};
      \node[actor] (d) at (6.25,0) {D$_1$};
      \node (dots2) at (7,0) {$\dots$};
      \node[actor] (d2) at (7.75,0) {D$_{n_2}$};
      \node[actor] (e) at (6.25,-0.75) {E$_1$};
      \node (dots3) at (7,-0.75) {$\dots$};
      \node[actor] (e2) at (7.75,-0.75) {E$_{n_3}$};
      \node[actor] (f) at (10.5,-2) {F \\ \emph{100 Hz}};
      \draw[arc] (a) to (dup-1);
      \draw[arc] (cj) to (f);
      \draw[arc] (dup-1) to (b);
      \draw[arc] (dup-1) to (cs);
      \draw[arc] (cs) to node[near start, sloped, above] {\textcolor{red}{1}} (c);
      \draw[arc] (cs) to node[near start, sloped, above] {\textcolor{red}{1}} (d);
      \draw[arc] (cs) to node[near start, sloped, above] {\textcolor{red}{0}} (e);
      \draw[arc] (c2) to node[near end, sloped, above] {\textcolor{red}{1}} (cj);
      \draw[arc] (d2) to node[near end, sloped, above] {\textcolor{red}{1}} (cj);
      \draw[arc] (e2) to node[near end, sloped, above] {\textcolor{red}{0}} (cj);
      \draw[arc,dashed] (b) to (dup-2);
      \draw[arc,dashed] (dup-2) to (cj);
      \draw[arc,dashed] (dup-2) to (cs);
    \end{tikzpicture}
  \end{subfigure}
  \caption{RMDF specifications that do not meet the required restrictions presented in \cref{prop:restrictions-control-areas}.}
\end{figure*}

\begin{definition}[Mode-coherent RMDF specification]
  Let $G$ be an RMDF specification. $G$ is mode-coherent if it satisfies the restrictions of \cref{prop:restrictions-control-areas}.
\end{definition}

\subsection{Consistency and Liveness Analysis of an RMDF Specification}

The consistency and liveness analysis of an RMDF specification is based on the static analysis of the exhaustive set of conditional execution branches of the RMDF specification. An RMDF specification with a single conditional execution branch can be analyzed as a routed PolyGraph specification. For instance, \cref{fig:yield-polygraph-specifications} presents the PolyGraph specifications obtained from an RMDF specification.

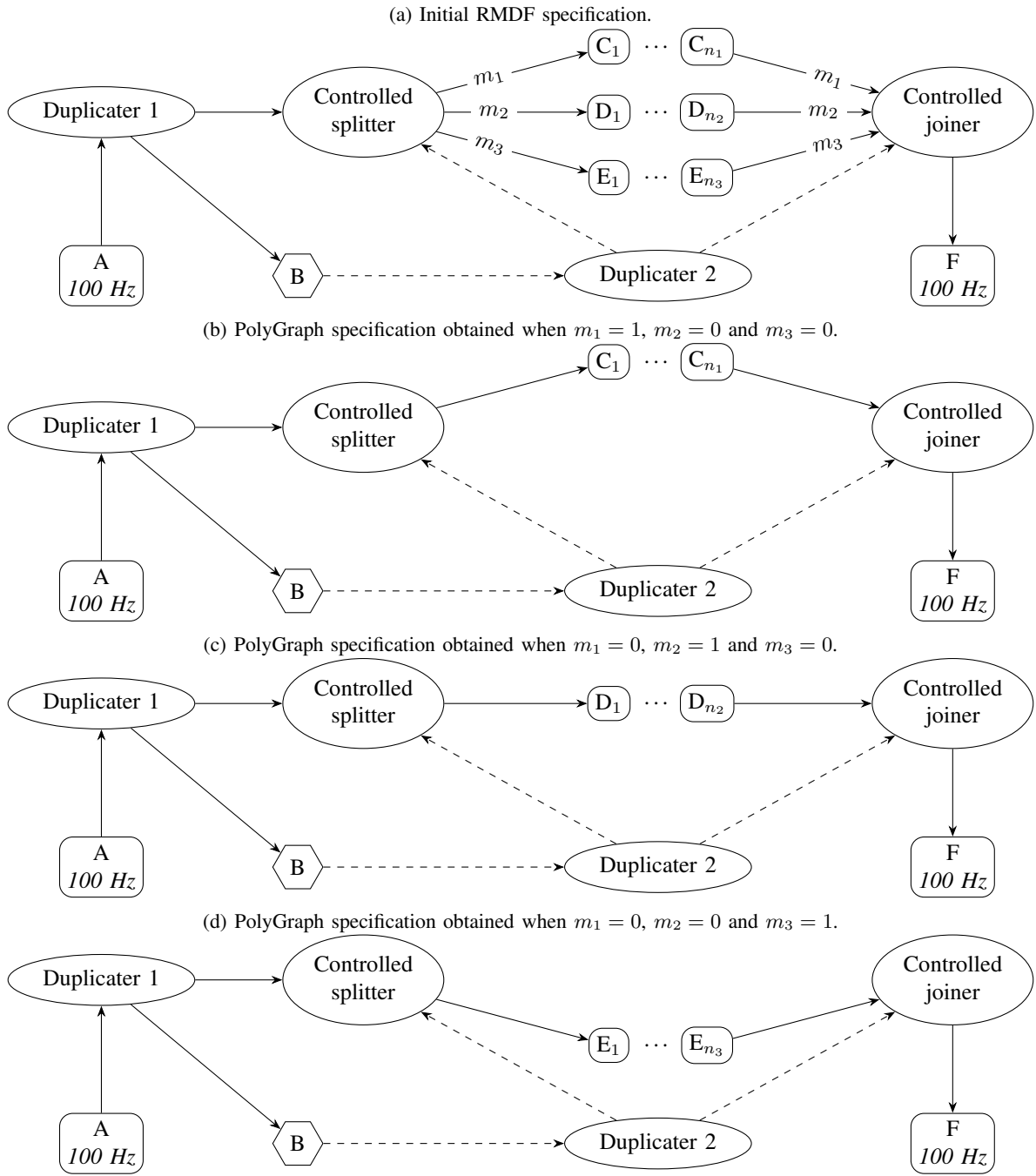
\begin{figure*}[htbp]
  \centering
  \begin{subfigure}{\textwidth}
    \centering
    \caption{Initial RMDF specification.}
    \begin{tikzpicture}
      \node[actor] (a) at (-1,-2.5) {A \\ \emph{100 Hz}};
      \node[virtualactor] (dup-1) at (-1,0) {Duplicater 1};
      \node[virtualactor] (dup-2) at (7.5,-2.5) {Duplicater 2};
      \node[virtualactor] (cs) at (3,0) {Controlled \\ splitter};
      \node[virtualactor] (cj) at (12,0) {Controlled \\ joiner};
      \node[decider] (b) at (2,-2.5) {B};
      \node[actor] (c1) at (6.75,1) {C$_1$};
      \node (ci) at (7.5,1) {\dots};
      \node[actor] (cl) at (8.25,1) {C$_{n_1}$};
      \node[actor] (d1) at (6.75,0) {D$_1$};
      \node (di) at (7.5,0) {\dots};
      \node[actor] (dm) at (8.25,0) {D$_{n_2}$};
      \node[actor] (e1) at (6.75,-1) {E$_1$};
      \node (ei) at (7.5,-1) {\dots};
      \node[actor] (en) at (8.25,-1) {E$_{n_3}$};
      \node[actor] (f) at (12,-2.5) {F \\ \emph{100 Hz}};
      \draw[arc] (a) to (dup-1);
      \draw[arc] (cj) to (f);
      \draw[arc] (dup-1) to (cs);
      \draw[arc] (dup-1) to (b);
      \draw[arc] (cs) to node[sloped,fill=white,pos=0.35] {$m_1$} (c1);
      \draw[arc] (cs) to node[sloped,fill=white,pos=0.35] {$m_2$} (d1);
      \draw[arc] (cs) to node[sloped,fill=white,pos=0.35] {$m_3$} (e1);
      \draw[arc] (en) to node[sloped,fill=white,pos=0.65] {$m_3$} (cj);
      \draw[arc] (dm) to node[sloped,fill=white,pos=0.65] {$m_2$} (cj);
      \draw[arc] (cl) to node[sloped,fill=white,pos=0.65] {$m_1$} (cj);
      \draw[arc,dashed] (b) to (dup-2);
      \draw[arc,dashed] (dup-2) to (cj);
      \draw[arc,dashed] (dup-2) to (cs);
    \end{tikzpicture}
  \end{subfigure}
  \begin{subfigure}{\textwidth}
    \centering
    \caption{PolyGraph specification obtained when $m_1 = 1$, $m_2 = 0$ and $m_3 = 0$.}
    \begin{tikzpicture}
      \node[actor] (a) at (-1,-2.5) {A \\ \emph{100 Hz}};
      \node[virtualactor] (dup-1) at (-1,0) {Duplicater 1};
      \node[virtualactor] (dup-2) at (7.5,-2.5) {Duplicater 2};
      \node[virtualactor] (cs) at (3,0) {Controlled \\ splitter};
      \node[virtualactor] (cj) at (12,0) {Controlled \\ joiner};
      \node[decider] (b) at (2,-2.5) {B};
      \node[actor] (c1) at (6.75,1) {C$_1$};
      \node (ci) at (7.5,1) {\dots};
      \node[actor] (cl) at (8.25,1) {C$_{n_1}$};
      \node[actor] (f) at (12,-2.5) {F \\ \emph{100 Hz}};
      \draw[arc] (a) to (dup-1);
      \draw[arc] (cj) to (f);
      \draw[arc] (dup-1) to (cs);
      \draw[arc] (dup-1) to (b);
      \draw[arc] (cs) to (c1);
      \draw[arc] (cl) to (cj);
      \draw[arc,dashed] (b) to (dup-2);
      \draw[arc,dashed] (dup-2) to (cj);
      \draw[arc,dashed] (dup-2) to (cs);
    \end{tikzpicture}
  \end{subfigure}
  \begin{subfigure}{\textwidth}
    \centering
    \caption{PolyGraph specification obtained when $m_1 = 0$, $m_2 = 1$ and $m_3 = 0$.}
    \begin{tikzpicture}
      \node[actor] (a) at (-1,-2.5) {A \\ \emph{100 Hz}};
      \node[virtualactor] (dup-1) at (-1,0) {Duplicater 1};
      \node[virtualactor] (dup-2) at (7.5,-2.5) {Duplicater 2};
      \node[virtualactor] (cs) at (3,0) {Controlled \\ splitter};
      \node[virtualactor] (cj) at (12,0) {Controlled \\ joiner};
      \node[decider] (b) at (2,-2.5) {B};
      \node[actor] (d1) at (6.75,0) {D$_1$};
      \node (di) at (7.5,0) {\dots};
      \node[actor] (dm) at (8.25,0) {D$_{n_2}$};
      \node[actor] (f) at (12,-2.5) {F \\ \emph{100 Hz}};
      \draw[arc] (a) to (dup-1);
      \draw[arc] (cj) to (f);
      \draw[arc] (dup-1) to (cs);
      \draw[arc] (dup-1) to (b);
      \draw[arc] (cs) to (d1);
      \draw[arc] (dm) to (cj);
      \draw[arc,dashed] (b) to (dup-2);
      \draw[arc,dashed] (dup-2) to (cj);
      \draw[arc,dashed] (dup-2) to (cs);
    \end{tikzpicture}
  \end{subfigure}
  \begin{subfigure}{\textwidth}
    \centering
    \caption{PolyGraph specification obtained when $m_1 = 0$, $m_2 = 0$ and $m_3 = 1$.}
    \begin{tikzpicture}
      \node[actor] (a) at (-1,-2.5) {A \\ \emph{100 Hz}};
      \node[virtualactor] (dup-1) at (-1,0) {Duplicater 1};
      \node[virtualactor] (dup-2) at (7.5,-2.5) {Duplicater 2};
      \node[virtualactor] (cs) at (3,0) {Controlled \\ splitter};
      \node[virtualactor] (cj) at (12,0) {Controlled \\ joiner};
      \node[decider] (b) at (2,-2.5) {B};
      \node[actor] (e1) at (6.75,-1) {E$_1$};
      \node (ei) at (7.5,-1) {\dots};
      \node[actor] (en) at (8.25,-1) {E$_{n_3}$};
      \node[actor] (f) at (12,-2.5) {F \\ \emph{100 Hz}};
      \draw[arc] (a) to (dup-1);
      \draw[arc] (cj) to (f);
      \draw[arc] (dup-1) to (cs);
      \draw[arc] (dup-1) to (b);
      \draw[arc] (cs) to (e1);
      \draw[arc] (en) to (cj);
      \draw[arc,dashed] (b) to (dup-2);
      \draw[arc,dashed] (dup-2) to (cj);
      \draw[arc,dashed] (dup-2) to (cs);
    \end{tikzpicture}
  \end{subfigure}
  \caption{An illustration of the PolyGraph specification obtained from an RMDF specification. All PolyGraph specifications are obtained by assigning values to $m_1$, $m_2$ and $m_3$.}
  \label{fig:yield-polygraph-specifications}
\end{figure*}

\begin{theorem}[Consistent and live RMDF specification]
  \label{theorem-1}
  Let $G = (V, E, M)$ be a mode-coherent RMDF specification. $G$ is consistent if the PolyGraph specification yield when parametric rates are assigned a value is consistent. $G$ is consistent and live if the PolyGraph specifications yield when parametric rates are assigned a value is consistent and live.
\end{theorem}

\begin{proof}[Proof of \cref{theorem-1}]
  Let $G = (V, E, M)$ be a mode-coherent RMDF specification. A symbolic execution of $G$ is a sequence of symbolic execution of PolyGraph specifications. Indeed, executions of mode deciders assign values to parameters, transforming an RMDF specification into a PolyGraph specification. If this latter is consistent and live, $G$ returns to its initial state without deadlocking. Note that the mode-coherence property of $G$ ensures that the execution of a control area is consistent and live. Once back to the initial state, the process repeats with another assignation of parameters, possibly yielding a different PolyGraph specification because another conditional execution branch can be taken. $G$ also returns to its initial state. It remains to verify that all PolyGraph specifications are consistent and live. Such techniques are defined in~\cite{dubrulle_data_2019}.
\end{proof}

\section{Timing Analysis of an RMDF Specification}

\label{sec:timing-analysis}

\subsection{Pre-processing of an RMDF Specification}

The RMDF model, as it is based on the routed PolyGraph model, allows system designers to define an arbitrary number of initial tokens on each channel. This feature is helpful for modeling data in memory, such as the initial state of a system, which is already available when the system starts. Therefore, an actor may have all its input channels filled with initial tokens, meaning that it can start its execution immediately. Such actor can undergo \enquote{offline execution} and be \emph{pre-processed} before the system starts without waiting for the first data produced by sensors. In other words, \enquote{offline jobs} may be executed to reduce the computational load of the system when it starts. It also eases the computation of the system's timing constraints, as we will see further in this paper.

The pre-processing of an RMDF specification is detailed in \cref{algo:pre-processing} (located on \cpageref{algo:pre-processing}), which in turn relies on \cref{algo:convertion-rational-rate}. This latter converts a rational rate into the number of tokens produced or consumed at a specific job instance. Indeed, a rational rate of a PolyGraph specification only gives a number of tokens produced or consumed over a set of consecutive jobs, but not at every job instance.

\cref{algo:pre-processing} pre-processes as many actors as possible using the initial tokens available in the channels. Actors with sufficient initial tokens in all their input channels are pre-processed. As each actor pre-processing produces tokens in their output channels, they may produce sufficient tokens for their successors to be pre-processed as well. Thus, each pre-processing must be followed by a check of the complete specification. \cref{algo:pre-processing} terminates when no jobs can be executed without depending on at least one job instance from a sensor (which jobs are always at runtime). In other words, all actors (except sensors) have a direct or indirect data dependency.

\begin{table}[htbp]
  \centering
  \caption{Notations and variables used in \cref{algo:pre-processing}.}
  \label{tab:notations-variables}
  \resizebox{\columnwidth}{!}{
    \begin{tabular}{|c|c|}
      \toprule
      \textbf{Notation/Variable} & \textbf{Description} \\
      \midrule
      $tkn\_state$ & \makecell{A data structure that stores tokens produced by the actors. \\ It is an array of queues with a FIFO policy. \\ The length of this array is equal to the number of channels \\ of the RMDF specification being pre-processed.} \\
      \midrule
      $ctrl\_tkn.channel$ & \makecell{The channel on which the controlled splitter/joiner that \\ consume control token will produce/consume its next token.} \\
      \bottomrule
    \end{tabular}
  }
\end{table}

\begin{algorithm}[htbp]
  \caption{Conversion of a rational rate to the number of tokens produced or consumed at the job level (Algorithm 2 in~\cite{dubrulle_data_2019}).}
  \label{algo:convertion-rational-rate}
  \KwIn{The $n$-th job instance of an actor $v \in V$ of a PolyGraph specification $G = (V, E)$, a channel $c \in E$ from/to which $v$ consumes/produces tokens with a rate $\gamma = \frac{p}{q} \in \mathbb{Q}^*$ and with $[c] \in \mathbb{Q}$ initial tokens.}
  \KwOut{The number of tokens consumed/produced by the $n$-th job of $v$ from/to $c$.}
  $r \leftarrow [c] - \lfloor [c] \rfloor$ \;
  \eIf{$\gamma > 0$}
  {
    \KwResult{$\lfloor n \cdot \gamma + r \rfloor - \lfloor (n - 1) \cdot \gamma + r \rfloor$}
  }
  {
    \KwResult{$\lceil n \cdot \gamma - r \rceil - \lceil (n - 1) \cdot \gamma - r \rceil$}
  }
\end{algorithm}

\begin{algorithm*}[htbp]
  \caption{Pre-processing of an RMDF specification.}
  \label{algo:pre-processing}
  \KwIn{An RMDF specification $G = (V, E, M)$.}
  \KwOut{The RMDF specification where all jobs have a direct or indirect successor with at least one sensor's job.}
  \Def {consume($ts, v, c, n$)}{
    \lForAll{$i \in [0, \dots, n-1]$}{$consumed\_token.get(i) \gets ts.get(v.c).dequeue()$}
    \KwResult{$consumed\_token$}
  }
  \Def {produce($ts, v, c, tkns$)}{
    \lForAll{$t \in tkns$}{$ts.get(v.c).enqueue(t)$}
  }
  $should\_restart \leftarrow \mathbf{True}$ ; $ execute\_actor \leftarrow \mathbf{True}$ ; $tkn\_state \gets []$ \;
  \lForAll{$v \in V$}{\textbf{if}~$v.is\_ctrl\_splitter~\mathbf{or}~v.is\_ctrl\_joiner$~\textbf{then}~$v.has\_consumed\_ctrl\_tkn \gets \mathbf{False}$}
  \While{$should\_restart$}
  {
    $should\_restart \gets \mathbf{False}$ \;
    \For{$v \in V$}
    {
      \If{$v.is\_mode\_decider~\mathbf{and}~ts.get(v.input\_channel.get(0)).size() \geq 1$}
      {
        \label{algo:line-1}
        $tkn \gets consume(ts, v, v.input\_channel.get(0), 1)$ \;
        $ctrl\_tkn.channel \gets v.decide\_mode(tkn)$ \;
        $produce(ts, v, v.ctrl\_channel, ctrl\_tkn)$ \;
        $should\_restart \gets \mathbf{True}$ \;
        \Continue
      }
      \If{$v.is\_ctrl\_splitter~\mathbf{or}~v.is\_ctrl\_joiner$}
      {
        \eIf{$v.has\_consumed\_ctrl\_tkn = \mathbf{False}~\textbf{and}~ts.get(v.ctrl\_channel).size() \geq 1$}
        {
          $ctrl\_tkn \gets consume(ts, v, v.ctrl\_channel, 1)$ \;
          $v.has\_consumed\_ctrl\_tkn \gets \mathbf{True}$ \;
          $should\_restart \gets \mathbf{True}$ \;
          \Continue
        }
        {
          \If{$ts.get(v.ctrl\_tkn.channel).size() \geq 1$}
          {
            \eIf{$v.is\_controlled\_splitter$}
            {
              $tkn \gets consume(ts, v, v.input\_channel.get(0), 1)$ \;
              $produce(ts, v, v.output\_channel.get(ctrl\_tkn.channel), tkn)$~;\\
              $v.has\_consumed\_ctrl\_tkn \gets \mathbf{False}$ \;
              $should\_restart \gets \mathbf{True}$ \;
              \Continue
            }
            {
              $tkn \gets consume(ts, v, v.input\_channel.get(ctrl\_tkn.channel), 1)$ \;
              $produce(ts, v, v.output\_channel.get(0), tkn)$~;\\
              $v.has\_consumed\_ctrl\_tkn \gets \mathbf{False}$ \;
              $should\_restart \gets \mathbf{True}$ \;
              \Continue
              \label{algo:line-2}
            }
          }
        }
      }
      \For{$c \in v.input\_channels$}
      {
        \lIf{$ts.get(c).size() < \cref{algo:convertion-rational-rate}(v.job\_instance, c, c.cons\_rate, [c])$}{$execute\_actor \gets \mathbf{False}$}
      }
      \If{$execute\_actor = \mathbf{True}$}
      {
        \For{$c \in v.input\_channels$}
        {
          $n \gets \cref{algo:convertion-rational-rate}(v.job\_instance, c, c.cons\_rate, [c])$ \;
          $v.consumed\_tkns.get(c).append(consume(ts, v, c, n))$~;\\
        }
        $v.produced\_tkns \gets execute\_actor(v)$ \;
        \lFor{$c \in v.output\_channels$}
        {
          $produce(ts, v, c, v.produced\_tkns.get(c))$
        }
        $v.job\_instance \gets v.job\_instance + 1$ \;
        $should\_restart \gets \mathbf{True}$ \;
      }
    }
  }
\end{algorithm*}

A consistent and live RMDF specification returns to its initial state after a hyperperiod. We claim that the timing constraints of a consistent, live and pre-processed RMDF specification (i.e., pre-processed with \cref{algo:pre-processing}) are cyclic over the first hyperperiod if the sensors and actuators of the specification have a frequency. Indeed, after pre-processing, all actors but the sensors have a direct or indirect data dependency with at least one sensor. Sensors are timed actors, so they release jobs periodically after the start of the execution. Consequently, they also periodically release their successors' jobs, ensuring cyclic timing constraints. The same reasoning applies to the actuators, which yield periodic deadline constraints that are propagated to their predecessor through data dependencies.

\subsection{Propagation of Timing Constraints}

The arithmetic relation of~\cite{roumage_static_2024}, tailored initially for a PolyGraph specification, can be applied to the RMDF specification. To do that, we must define a similar notion as \emph{well-defined PolyGraph specification} of~\cite{roumage_static_2024} for an RMDF specification.

\begin{definition}[Well-defined PolyGraph/RMDF specification]
  A well-defined PolyGraph/RMDF specification is a weakly connected graph\footnote{A directed graph is weakly connected when its undirected induced graph is connected~\cite{honorat_scheduling_2020}.} that is both consistent and live. Actors without predecessors/successors (i.e., the sensors/actuators) are timed actors, meaning they have a frequency constraint. All actors within the specification have both a BCET and a WCET, such that the BCET is less than or equal to the WCET. A well-defined PolyGraph/RMDF specification is also \emph{pre-processed}, meaning that it has been processed with \cref{algo:pre-processing}. Note that for a PolyGraph specification, \cref{algo:line-1} to \cref{algo:line-2} of \cref{algo:pre-processing} are not executed. 
\end{definition}

We remind the main result of~\cite{roumage_static_2024} that allows the derivation of the execution windows of a PolyGraph specification. This result will be extended to an RMDF specification.

\begin{theorem}[Derivation of the execution windows of a PolyGraph specification - Theorem 1 of~\cite{roumage_static_2024}]
  \label{thm:execution-windows-polygraph}
  Let $G = (V, E)$ be a well-defined PolyGraph specification, let $v_j, v_k \in V$ be two actors of $G$, and let $c_i \in E$ be a channel of $G$ from $v_j$ to $v_k$. We define $G_\Gamma = (\gamma_{ij}) \in \mathbb{Q}^{|E| \times |V|}$ to be the topology matrix of $G$, and $r_i = [c_i] - \lfloor [c_i] \rfloor$ with $\lfloor [c_i] \rfloor$ being the floor function applied to $[c_i]$. We define $\texttt{period}_v$ and $\texttt{phase}_v$ to be the period and phase of an actor $v \in V$ if $v$ is a timed actor. Let $bcet_v/wcet_v$ be the BCET/WCET of an actor $v \in V$. Then, the following holds:
  \begin{enumerate}
    \item Let $v_k \in V$ be a non-timed actor of $G$. Then, assuming $release_{v_j}^n$ returns the release of the $n$-th job of an actor $v_j \in V$, the release of the $p$-th job of $v_k$ is computed as follows:
    \begin{equation*}
      \label{eq:release-non-timed-actor}
        \begin{split}
          release_{v_k}^p = \max\limits_{\substack{i \in [0, |E|] \\ j \in [0, |V|] \\ \text{st } \gamma_{ij} < 0}} & release_{v_j}^{\alpha_1 (p, i, j, k)} + bcet_{v_j} \\ & + \bigl( p - \beta_1 (p, i, j, k) \bigr) \cdot bcet_{v_k}
        \end{split}
    \end{equation*}
    with
    $$\alpha_1 (p, i, j, k) =  \bigl \lceil \frac{\lceil p \cdot |\gamma_{ik}| - r_i \rceil - [c_i]}{\gamma_{ij}} \bigr \rceil $$
    and
    $$\beta_1 (p, i, j, k) = 1 + \Bigl \lfloor \frac{\lfloor \bigl( \alpha_1 (p, i, j, k) - 1 \bigr) \cdot \gamma_{ij} + [c_i] \rfloor + r_i}{|\gamma_{ik}|} \Bigr \rfloor$$
    \item Let $v_k \in V$ be a timed actor of $G$. Then, by denoting $\overline{release}_{v_k}^p$ the release computed by \cref{eq:release-non-timed-actor}, the release of the $p$-th job of $v_k$ is computed as follows:
    {\small
    \begin{equation*}
      \label{eq:release-timed-actor}
      release_{v_k}^p = \max \Bigl( \overline{release}_{v_k}^p~,~\texttt{period}_{v_k} \cdot (p - 1) + \texttt{phase}_{v_k} \Bigr)
    \end{equation*}
    }
    \item Let $v_j \in V$ be a non-timed actor of $G$. Then, assuming $deadline_{v_k}^p$ returns the deadline of the $p$-th job of an actor $v_k \in V$, the deadline of the $n$-th job of $v_j$ is computed as follows:
    \begin{equation*}
      \label{eq:deadline-non-timed-actor}
        \begin{split}
          deadline_{v_j}^n = \min\limits_{\substack{i \in [0, |E|] \\ k \in [0, |V|] \\ \text{st } \gamma_{ik} > 0}} & deadline_{v_k}^{\alpha_2 (n, i, j, k)} - wcet_{v_k} \\ & - \Bigl( \beta_2 (n, i, j, k) - n \Bigr) \cdot wcet_{v_j}
        \end{split}
    \end{equation*}
    with
    \begin{equation*}
        \alpha_2 (n, i, j, k) = 1 + \Bigl \lfloor \frac{\lfloor (n - 1) \cdot \gamma_{ij} + [c_i] \rfloor + r_i}{|\gamma_{ik}|} \Bigr \rfloor
    \end{equation*}
    and
    \begin{equation*}
        \beta_2 (n, i, j, k) = \bigl \lceil \frac{\lceil \alpha_2 (n, i, j, k) \cdot |\gamma_{ik}| - r_i \rceil - [c_i]}{\gamma_{ij}} \bigr \rceil
    \end{equation*}
    \item Let $v_j \in V$ be a timed actor of $G$. Then, by denoting $\overline{deadline}_{v_j}^n$ the deadline computed by \cref{eq:deadline-non-timed-actor}, the deadline of the $n$-th job of $v_j$ is computed as follows:
    {\small
    \begin{equation*}
      \label{eq:deadline-timed-actor}
      deadline_{v_j}^n = \min \Bigl( \overline{deadline}_{v_j}^n~,~\texttt{period}_{v_j} \cdot n + \texttt{phase}_{v_j} \Bigr)
    \end{equation*}
    }
  \end{enumerate}
\end{theorem}

The interested reader can find the proof of \cref{thm:execution-windows-polygraph} in~\cite{roumage_static_2024}.

\begin{theorem}[Derivation of the execution windows of an RMDF specification]
  \label{thm:execution-windows-rmdf}
  Let $G = (V, E, M)$ be an RMDF specification where all parameters are set to 1. Under the assumptions of \cref{thm:execution-windows-polygraph}, except that $G$ is a well-defined RMDF specification instead of a well-defined PolyGraph specification, the result of \cref{thm:execution-windows-polygraph} can be extended to $G$.
\end{theorem}

\begin{proof}[Proof of \cref{thm:execution-windows-rmdf}]
  An RMDF specification as defined in the assumptions of \cref{thm:execution-windows-polygraph}, i.e., with all its parameters set to 1, can be analyzed as a well-defined PolyGraph specification, which is the assumption of \cref{thm:execution-windows-polygraph}.
\end{proof}

\cref{thm:execution-windows-rmdf} extends the timing constraints propagation of a PolyGraph specification to an RMDF specification by considering all parameters to be set to 1. This implies a \emph{conservative approach} for the execution windows assignment of the controlled splitters and controlled joiners. Deadlines are set to the minimum of the deadlines of their output branches. Similarly, the releases are set to the maximum of the releases of their input branches.

\section{A Feasibility Test of an RMDF Specification}

\label{sec:feasibility-test}

The timing constraint propagation described in the previous section provides a timing assessment of an RMDF specification, serving as a feasibility test. Let $G = (V, E)$ be a well-defined RMDF specification and $G_{rv}(v)$ the repetition vector of $G$, that is, the number of executions of $v$ in a hyperperiod of $G$. A \emph{necessary} condition for the system to be feasible is to assert if the execution window length of all jobs is greater or equal to the WCET of the actor:

$$
  \forall v \in V, \forall i \in G_{rv} (v) : wcet_v \leq deadline_v^i - release_v^i
$$

The proposed timing assessment reduces the domain search space of dataflow specifications, as unfeasible RMDF specifications can be detected as early as possible. \cref{fig:design-flow} illustrates the design flow we proposed with RMDF.

\begin{figure*}[htbp]
  \centering
  \resizebox{\textwidth}{!}{
    \begin{tikzpicture}
      \node[tapeactor,fill=gray!50] (spec) at (-0.5,0) {Specifications of a \\ mode-dependent \\ system with relaxed \\ real-time constraints};
      \node[actor] (model-hw-independent) at (5,0) {RMDF specification \\ of that system};
      \node[actor] (static-analysis) at (5,-2) {Consistency and liveness test};
      \node[actor] (model-hw-dependent) at (5,-4) {RMDF specifications \\ with BCETs and WCETs};
      \node[tapeactor,fill=gray!50] (bcet-wcet) at (-0.5,-4) {BCETs and \\ WCETs};
      \node[diamond,align=center,draw] (analysis) at (11,-3) {Feasibility \\ test \\ (cf. \cref{sec:feasibility-test})};
      \node[actor] (timing-constraints) at (8,-6) {Computation of \\ timing constraints \\ at the job level};
      \node[cloud,align=center,draw,fill=gray!50] (toto) at (15.5,-7) {Scheduling \\ Monitoring \\ Sizing \\ Etc};
      \draw[arc] (model-hw-dependent) to[out=270,in=180] (timing-constraints);
      \draw[arc] (timing-constraints) to[out=0,in=270] (analysis);
      \draw[arc] (bcet-wcet) -- (model-hw-dependent);
      \draw[arc] (model-hw-independent) -- (static-analysis);
      \draw[arc] (static-analysis) -- (model-hw-dependent);
      \draw[arc] (analysis) to[out=90,in=0] node[pos=0.7,above] {No Valid: refinement} (model-hw-independent);
      \draw[arc] (spec) -- (model-hw-independent);
      \draw[arc] (analysis) to[out=0,in=90] node[midway,align=center,right] {Valid: further \\ analysis} (toto);
      \node[draw, rectangle, rounded corners, dashed, ultra thick, fit=(model-hw-independent) (analysis) (timing-constraints) (model-hw-dependent) (static-analysis)] {};
    \end{tikzpicture}
  }
  \caption{A design flow with RMDF. The central dotted rectangle is the workflow we proposed.}
  \label{fig:design-flow}
\end{figure*}
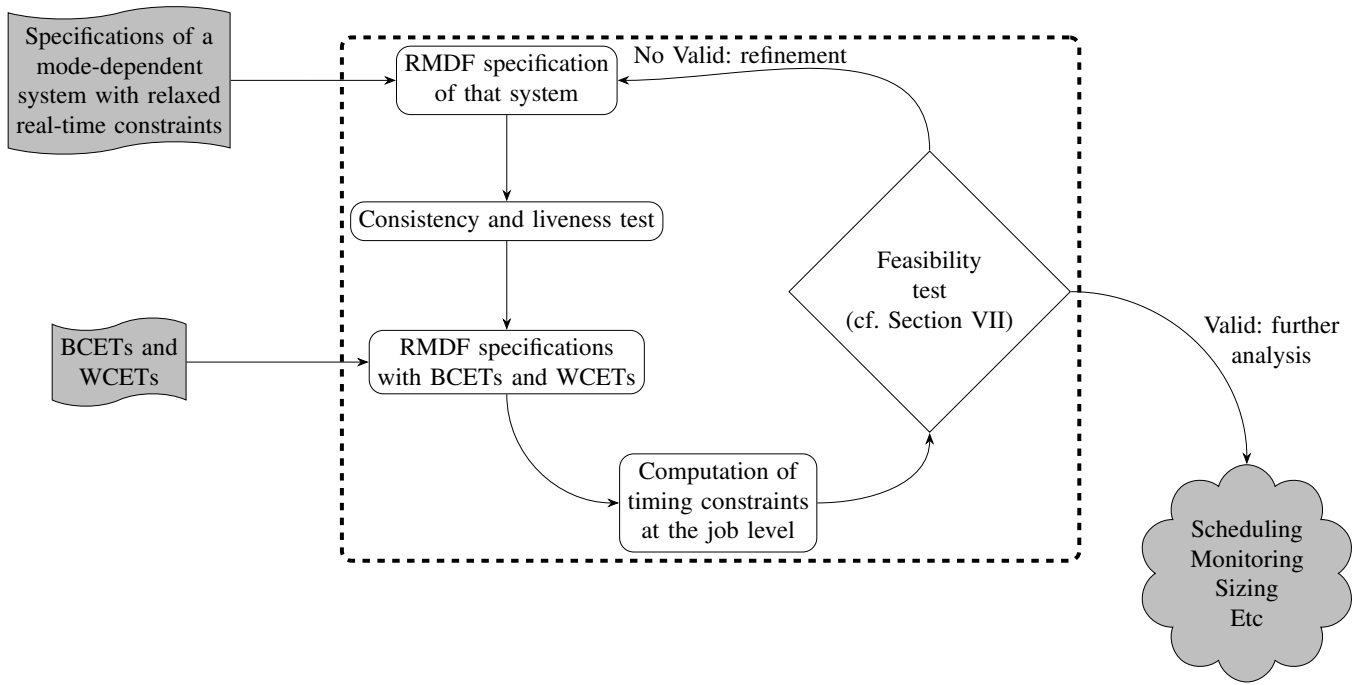

\section{Application to the Vision System of Ingenuity}

\label{sec:applications}

As a reminder, a dataflow model is a trade-off between its expressiveness and analyzability. In \cref{sec:rmdf}, we have created the RMDF model, which extends the expressiveness of PolyGraph while maintaining its consistency and liveness analyses as detailed in \cref{sec:static-analysis}. In \cref{sec:timing-analysis} and \cref{sec:feasibility-test}, the analyzability of RMDF was extended while maintaining its expressiveness. The current section applies the previous results to a dataflow specification of the vision processing system of Ingenuity.

\subsection{An Enhanced Model Fidelity with RMDF}

We have specified the vision processing system of the Ingenuity Mars helicopter~\cite{bayard_visionbased_2019} in \cref{fig:specification-ingenuity-vision-processing} with both the PolyGraph model (\cref{fig:specification-ingenuity-vision-processing-1}) and the RMDF model (\cref{fig:specification-ingenuity-vision-processing-2}). The vision processing system detects visual features on camera frames and either stores them as reference features (modeled as \emph{Pseudo landmarks}) or tracks them from one frame to another. After a filtering procedure, the tracked visual features are matched with pseudo landmarks, enabling the computation of a visual drift to enhance position estimation.

New pseudo landmarks are generated if certain conditions are not met, and some of those conditions depend on runtime data. Thus, we have to simplify the pseudo landmarks generation procedure in the PolyGraph specification and assume that 1 image over 10 is used to generate new pseudo landmarks\footnote{In the real system, at least 1 image over 10 is used to generate pseudo landmarks. This condition is not the only one used in the real system. However, it is the only one that can be modeled statically. All other conditions depend on runtime data~\cite{bayard_visionbased_2019}.}.

The expressiveness of the RMDF model allows us to model the pseudo landmark generation with a mode decider actor (which is \emph{Label decider} in \cref{fig:specification-ingenuity-vision-processing-2}), which is more accurate than the simplification made in the PolyGraph specification. Indeed, the selection algorithm inside the \emph{Label decider} may depend on runtime data, which is impossible to specify in the PolyGraph specification.

\begin{figure}
  \centering
  \begin{subfigure}{\columnwidth}
    \caption{A PolyGraph specification of the vision processing system of Ingenuity.}
    \label{fig:specification-ingenuity-vision-processing-1}
    \resizebox{\columnwidth}{!}{
      \begin{tikzpicture}
        \node[actor] (cam) at (0,0) {Camera \\ \emph{30 Hz}};
        \node[actor] (fd) at (0,-2) {Feature \\ detection};
        \node[virtualactor] (spl) at (2,0) {Splitter};
        \node[actor] (pl) at (2,-2) {Pseudo \\ landmarks};
        \node[actor] (ft) at (5,0) {Feature \\ tracking};
        \node[virtualactor] (join) at (5,-2) {Joiner};
        \node[actor] (or) at (7,0) {Filtering \\ procedure};
        \node[actor] (fm) at (7,-2) {Feature \\ match \\ \emph{30 Hz}};
        \draw[arc] (cam) to (fd);
        \draw[arc] (fd) to (spl);
        \draw[arc] (spl) to node[very near start, right] {1/10} node[midway, right] {[9/10]} (pl);
        \draw[arc] (spl) to node[near start, above] {9/10} (ft);
        \draw[arc] (pl) to node[near end, below] {1/10} (join);
        \draw[arc] (ft) to (or);
        \draw[arc] (or) to node[midway, right] {[9/10]} node[near end, left] {9/10} (join);
        \draw[arc] (join) to (fm);
        \draw[arc] (ft.60) to[out=110,in=70] node[midway,above] {[1]} (ft.130);
      \end{tikzpicture}
    }
  \end{subfigure}
  \begin{subfigure}{\columnwidth}
    \caption{An RMDF specification of the vision processing system of Ingenuity.}
    \label{fig:specification-ingenuity-vision-processing-2}
    \resizebox{\columnwidth}{!}{
      \begin{tikzpicture}
        \node[actor] (cam) at (-1,0) {Camera \\ \emph{30 Hz}};
        \node[actor] (fd) at (-1,-3) {Feature \\ detection};
        \node[virtualactor] (dup) at (-1,-6) {Duplicater 1};
        \node[virtualactor] (spl) at (2,0) {Controlled \\ splitter};
        \node[actor] (pl) at (2,-3) {Pseudo \\ landmarks};
        \node[decider] (lab) at (2,-6) {Label \\ decider};
        \node[virtualactor] (dup2) at (5.5,-6) {Duplicater 2};
        \node[actor] (ft) at (5.5,0) {Feature \\ tracking};
        \node[virtualactor] (join) at (5.5,-3) {Controlled \\ joiner};
        \node[actor] (or) at (8,0) {Filtering \\ procedure};
        \node[actor] (fm) at (8,-3) {Feature \\ match \\ \emph{30 Hz}};
        \draw[arc] (cam) to (fd);
        \draw[arc] (fd) to (dup);
        \draw[arc] (dup) to (lab);
        \draw[arc,dashed] (dup2) .. controls (0,-3.5) .. (spl);
        \draw[arc,dashed] (dup2) to (join);
        \draw[arc,dashed] (lab) to (dup2);
        \draw[arc] (spl) to node[very near start, right] {$m_1$} (pl);
        \draw[arc] (spl) to node[near start, above] {$m_2$} (ft);
        \draw[arc] (pl) to node[near end, below] {$m_1$} (join);
        \draw[arc] (ft) to (or);
        \draw[arc] (or) to node[very near end, above, sloped] {$m_2$} (join);
        \draw[arc] (join) to (fm);
        \draw[arc] (ft.60) to[out=110,in=70] node[midway,above] {[1]} (ft.130);
        \draw[arc] (dup.30) to[out=60,in=240] (spl.200);
      \end{tikzpicture}
    }
  \end{subfigure}
  \begin{subfigure}{\columnwidth}
    \caption{A modified RMDF specification of the vision processing system of Ingenuity to have more interesting result for the timing constraints propagation.}
    \label{fig:specification-ingenuity-vision-processing-3}
    \resizebox{\columnwidth}{!}{
      \begin{tikzpicture}
        \node[actor] (cam) at (-1,0) {Camera \\ \emph{30 Hz}};
        \node[actor] (fd) at (-1,-3) {Feature \\ detection};
        \node[virtualactor] (dup) at (-1,-6) {Duplicater 1};
        \node[virtualactor] (spl) at (2,0) {Controlled \\ splitter};
        \node[actor] (pl) at (2,-3) {Pseudo \\ landmarks};
        \node[decider] (lab) at (2,-6) {Label \\ decider};
        \node[virtualactor] (dup2) at (5.5,-6) {Duplicater 2};
        \node[actor] (ft) at (5.5,0) {Feature \\ tracking};
        \node[virtualactor] (join) at (5.5,-3) {Controlled \\ joiner};
        \node[actor] (or) at (8,0) {Filtering \\ procedure};
        \node[actor] (fm) at (8,-3) {Feature \\ match};
        \node[actor] (mot) at (8,-6) {Motors \\ \emph{500 Hz}};
        \draw[arc] (cam) to (fd);
        \draw[arc] (fd) to (dup);
        \draw[arc] (dup) to (lab);
        \draw[arc,dashed] (dup2) .. controls (0,-3.5) .. (spl);
        \draw[arc,dashed] (dup2) to (join);
        \draw[arc,dashed] (lab) to (dup2);
        \draw[arc] (spl) to node[very near start, right] {$m_1$} (pl);
        \draw[arc] (spl) to node[near start, above] {$m_2$} (ft);
        \draw[arc] (pl) to node[near end, below] {$m_1$} (join);
        \draw[arc] (ft) to (or);
        \draw[arc] (or) to node[very near end, above, sloped] {$m_2$} (join);
        \draw[arc] (join) to (fm);
        \draw[arc] (ft.60) to[out=110,in=70] node[midway,above] {[1]} (ft.130);
        \draw[arc] (dup.30) to[out=60,in=240] (spl.200);
        \draw[arc] (fm) to node[midway,left] {[1/50]} node[very near end,left] {3/50} (mot);
      \end{tikzpicture}
    }
  \end{subfigure}
  \caption{Specifications of the Ingenuity vision processing system with PolyGraph and RMDF. All those specifications are consistent and live. Rates of 1 and initial tokens of 0 are omitted for clarity.}
  \label{fig:specification-ingenuity-vision-processing}
\end{figure}

\subsection{Lazy Evaluation Propagation of Timing Constraints}

The timing propagation presented in~\cite{roumage_static_2024} and reminded in the previous section applies to an RMDF specification. Let us apply it to a slightly modified RMDF specification of the Ingenuity vision processing system of \cref{fig:specification-ingenuity-vision-processing-2}. This modification, presented in \cref{fig:specification-ingenuity-vision-processing-3}, is the following: 1) remove the frequency of \emph{Feature Match}, 2) add an actor \emph{Motors} with a frequency of 500 Hz and 3) add a channel between \emph{Feature Match} and \emph{Motors} with a production rate of 1, a consumption rate of 3/50 and 1/50 initial token. This modification resulted in a more interesting outcome for the timing propagation, as there is a fractional rate (the non-modified specification only has rates of 1). Note that the actor \emph{Motors} specified the rest of the Ingenuity system, especially the motors triggering the helicopter's movements at a frequency of 500 Hz, as defined by the textual specification~\cite{grip_flight_2019}. The results of the propagation of timing constraints is presented in \cref{tab:timing-propagation}, in which we assumed that BCETs and WCETs of all actors are 0.12 and 0.20 ms, respectively.

The results of \cref{thm:execution-windows-polygraph} and its extension to RMDF allow a timing constraints propagation with a lazy evaluation. Specifically, arithmetic formulas enable us to compute directly the job of interest, e.g., the successor job that consumes the token produced by an actor. As a result, only the necessary and sufficient computations are performed to derive the timing constraints. This lazy evaluation is particularly useful when timing constraints are computed for a subset of actors or even a subset of jobs.

\begin{table*}
  \caption{Releases, deadlines and execution windows length of the $n$-th job for each actor of the Ingenuity vision processing system, assuming BCET and WCET of all actors are 0.12 and 0.20 ms, respectively, except for the \emph{Controlled Splitter} and \emph{Controlled Joiner} which have constant execution time of 0 ms.}
  \centering
  \label{tab:timing-propagation}
  \begin{tabular}{cccc}
    \toprule
    \multirow{2}{*}{\textbf{Actor}} & \multicolumn{3}{c}{\textbf{Timing constraint of the $n$-th job}} \\
    & \textbf{Release} & \textbf{Deadline} & \textbf{Execution window} \\
    \midrule
    Camera & $\begin{cases} 100 \cdot \lfloor (n-1) / 3 \rfloor & \text{ if } (n-1) \mod{3} = 0 \\ 100/3 + 100 \cdot \lfloor (n-1) / 3 \rfloor & \text{ if } (n-1) \mod{3} = 1 \\ 200/3 + 100 \cdot \lfloor (n-1) / 3 \rfloor & \text{ if } (n-1) \mod{3} = 2 \end{cases}$ & $\begin{cases} 7/5 + 100 \cdot \lfloor (n-1) / 3 \rfloor \\ 177/5 + 100 \cdot \lfloor (n-1) / 3 \rfloor \\ 337/5 + 100 \cdot \lfloor (n-1) / 3 \rfloor \end{cases}$ & $\begin{cases} 7/5 = 1.4 & \text{ ms} \\ 31/15 \simeq 2.1 & \text{ ms} \\ 11/15 \simeq 0.73 & \text{ ms} \end{cases}$ \\
    \midrule
    Feature Detection & $\begin{cases} 3/25 + 100 \cdot \lfloor (n-1) / 3 \rfloor & \text{ if } (n-1) \mod{3} = 0 \\ 2509/75 + 100 \cdot \lfloor (n-1) / 3 \rfloor & \text{ if } (n-1) \mod{3} = 1 \\ 5009/75 + 100 \cdot \lfloor (n-1) / 3 \rfloor & \text{ if } (n-1) \mod{3} = 2 \end{cases}$ & $\begin{cases} 8/5 + 100 \cdot \lfloor (n-1) / 3 \rfloor \\ 178/5 + 100 \cdot \lfloor (n-1) / 3 \rfloor \\ 338/5 + 100 \cdot \lfloor (n-1) / 3 \rfloor \end{cases}$ & $\begin{cases} 37/25 = 1.48 & \text{ ms} \\ 161/75 \simeq 2.1 & \text{ ms} \\ 61/75 \simeq 0.81 & \text{ ms} \end{cases}$ \\
    \midrule
    Label Decider & $\begin{cases} 6/25 + 100 \cdot \lfloor (n-1) / 3 \rfloor & \text{ if } (n-1) \mod{3} = 0 \\ 2518/75 + 100 \cdot \lfloor (n-1) / 3 \rfloor & \text{ if } (n-1) \mod{3} = 1 \\ 5018/75 + 100 \cdot \lfloor (n-1) / 3 \rfloor & \text{ if } (n-1) \mod{3} = 2 \end{cases}$ & $\begin{cases} 9/5 + 100 \cdot \lfloor (n-1) / 3 \rfloor \\ 179/5 + 100 \cdot \lfloor (n-1) / 3 \rfloor \\ 339/5 + 100 \cdot \lfloor (n-1) / 3 \rfloor \end{cases}$ & $\begin{cases} 39/25 = 1.56 & \text{ ms} \\ 167/75 \simeq 2.2 & \text{ ms} \\ 67/75 \simeq 0.87 & \text{ ms} \end{cases}$ \\
    \midrule
    Controlled Splitter & $\begin{cases} 9/25 + 100 \cdot \lfloor (n-1) / 3 \rfloor & \text{ if } (n-1) \mod{3} = 0 \\ 2527/75 + 100 \cdot \lfloor (n-1) / 3 \rfloor & \text{ if } (n-1) \mod{3} = 1 \\ 5027/75 + 100 \cdot \lfloor (n-1) / 3 \rfloor & \text{ if } (n-1) \mod{3} = 2 \end{cases}$ & $\begin{cases} 2 + 100 \cdot \lfloor (n-1) / 3 \rfloor \\ 36 + 100 \cdot \lfloor (n-1) / 3 \rfloor \\ 68 + 100 \cdot \lfloor (n-1) / 3 \rfloor \end{cases}$ & $\begin{cases} 41/25 = 1.64 & \text{ ms} \\ 173/75 \simeq 2.3 & \text{ ms} \\ 73/75 \simeq 0.97 & \text{ ms} \end{cases}$ \\
    \midrule
    Feature Tracking & $\begin{cases} 12/25 + 100 \cdot \lfloor (n-1) / 3 \rfloor & \text{ if } (n-1) \mod{3} = 0 \\ 2536/75 + 100 \cdot \lfloor (n-1) / 3 \rfloor & \text{ if } (n-1) \mod{3} = 1 \\ 5036/75 + 100 \cdot \lfloor (n-1) / 3 \rfloor & \text{ if } (n-1) \mod{3} = 2 \end{cases}$ & $\begin{cases} 11/5 + 100 \cdot \lfloor (n-1) / 3 \rfloor \\ 181/5 + 100 \cdot \lfloor (n-1) / 3 \rfloor \\ 341/5 + 100 \cdot \lfloor (n-1) / 3 \rfloor \end{cases}$ & $\begin{cases} 43/25 = 1.72 & \text{ ms} \\ 179/75 \simeq 2.4 & \text{ ms} \\ 79/75 \simeq 1.0 & \text{ ms} \end{cases}$ \\
    \midrule
    Filtring Procedure & $\begin{cases} 3/5 + 100 \cdot \lfloor (n-1) / 3 \rfloor & \text{ if } (n-1) \mod{3} = 0 \\ 509/15 + 100 \cdot \lfloor (n-1) / 3 \rfloor & \text{ if } (n-1) \mod{3} = 1 \\ 1009/15 + 100 \cdot \lfloor (n-1) / 3 \rfloor & \text{ if } (n-1) \mod{3} = 2 \end{cases}$ & $\begin{cases} 12/5 + 100 \cdot \lfloor (n-1) / 3 \rfloor \\ 182/5 + 100 \cdot \lfloor (n-1) / 3 \rfloor \\ 342/5 + 100 \cdot \lfloor (n-1) / 3 \rfloor \end{cases}$ & $\begin{cases} 9/5 = 1.8 & \text{ ms} \\ 37/15 \simeq 2.5 & \text{ ms} \\ 17/15 \simeq 1.1 & \text{ ms} \end{cases}$ \\
    \midrule
    Pseudo Landmarks & $\begin{cases} 12/25 + 100 \cdot \lfloor (n-1) / 3 \rfloor & \text{ if } (n-1) \mod{3} = 0 \\ 2536/75 + 100 \cdot \lfloor (n-1) / 3 \rfloor & \text{ if } (n-1) \mod{3} = 1 \\ 5036/75 + 100 \cdot \lfloor (n-1) / 3 \rfloor & \text{ if } (n-1) \mod{3} = 2 \end{cases}$ & $\begin{cases} 12/5 + 100 \cdot \lfloor (n-1) / 3 \rfloor \\ 182/5 + 100 \cdot \lfloor (n-1) / 3 \rfloor \\ 342/5 + 100 \cdot \lfloor (n-1) / 3 \rfloor \end{cases}$ & $\begin{cases} 48/25 = 1.92 & \text{ ms} \\ 194/75 \simeq 2.6 & \text{ ms} \\ 94/75 \simeq 1.2 & \text{ ms} \end{cases}$ \\
    \midrule
    Controlled Splitter & $\begin{cases} 18/25 + 100 \cdot \lfloor (n-1) / 3 \rfloor & \text{ if } (n-1) \mod{3} = 0 \\ 2554/75 + 100 \cdot \lfloor (n-1) / 3 \rfloor & \text{ if } (n-1) \mod{3} = 1 \\ 5054/75 + 100 \cdot \lfloor (n-1) / 3 \rfloor & \text{ if } (n-1) \mod{3} = 2 \end{cases}$ & $\begin{cases} 13/5 + 100 \cdot \lfloor (n-1) / 3 \rfloor \\ 183/5 + 100 \cdot \lfloor (n-1) / 3 \rfloor \\ 343/5 + 100 \cdot \lfloor (n-1) / 3 \rfloor \end{cases}$ & $\begin{cases} 47/25 = 1.88 & \text{ ms} \\ 191/75 \simeq 2.5 & \text{ ms} \\ 91/75 \simeq 1.2 & \text{ ms} \end{cases}$ \\
    \midrule
    Feature Match & $\begin{cases} 21/25 + 100 \cdot \lfloor (n-1) / 3 \rfloor & \text{ if } (n-1) \mod{3} = 0 \\ 2563/75 + 100 \cdot \lfloor (n-1) / 3 \rfloor & \text{ if } (n-1) \mod{3} = 1 \\ 5063/75 + 100 \cdot \lfloor (n-1) / 3 \rfloor & \text{ if } (n-1) \mod{3} = 2 \end{cases}$ & $\begin{cases} 14/5 + 100 \cdot \lfloor (n-1) / 3 \rfloor \\ 184/5 + 100 \cdot \lfloor (n-1) / 3 \rfloor \\ 344/5 + 100 \cdot \lfloor (n-1) / 3 \rfloor \end{cases}$ & $\begin{cases} 49/25 = 1.96 & \text{ ms} \\ 197/75 \simeq 2.6 & \text{ ms} \\ 97/75 \simeq 1.3 & \text{ ms} \end{cases}$ \\
    \midrule
    Motors & $1 + 2 \cdot (n-1)$ & $1 + 2 \cdot n$ & $2 \text{ ms}$ \\
    \bottomrule
  \end{tabular}
\end{table*}

\subsection{A Feasibility Test}

To put into context the feasibility test presented in section \cref{sec:feasibility-test}, the result of \cref{tab:timing-propagation} can be interpreted as follows. It is necessary that the WCET of all actors of Ingenuity are less than the minimum execution window length presented in column \emph{Execution window} of \cref{tab:timing-propagation} for Ingenuity to be feasible. Those values are summarized in \cref{tab:wcet}.

A schedulability test has been proposed for the PolyGraph model in~\cite{hamelin_performance_2024}. While in this paper we do not consider a scheduling policy, the authors of~\cite{hamelin_performance_2024} consider a preemptive context and interferences between actors. A response time analysis is performed to determine the worst-case response time of actors, considering the worst-case execution time of actors and the interferences between actors. Another difference is the timing evaluation procedure. While our approach is a lazy evaluation, the approach of~\cite{hamelin_performance_2024} recursively updates the response time of actors until a fixed point is reached. Both approaches achieve the same result, and the lazy evaluation requires less computation if all timing constraints are unnecessary.

\begin{table}[htbp]
  \caption{Maximum WCETs of the actors of the Ingenuity vision processing system that constitute a necessary condition for its feasibility.}
  \centering
  \label{tab:wcet}
  \begin{tabular}{|cc|cc|}
    \toprule
    \textbf{Actors} & \textbf{WCET} & \textbf{Actors} & \textbf{WCET} \\
    \midrule
    Camera & 0.73 ms & Filtering Procedure & 1.1 ms \\
    Feature Detection & 0.81 ms & Pseudo Landmarks & 1.2 ms \\
    Label Decider & 0.87 ms & Controlled Splitter & 1.2 ms \\
    Controlled Splitter & 1.2 ms & Feature Match & 1.3 ms \\
    Feature Tracking & 1.0 ms & Motors & 2 ms \\
    \bottomrule
  \end{tabular}
\end{table}

\section{The Mode Change Protocol of RMDF}

\label{sec:mode-change-protocol}

Mode-oriented execution relies on \emph{Mode Change Protocols} (MCP)~\cite{fort_synchronous_2022}. A \emph{Mode Change Protocol} defines how the transition from one mode to another is handled. A \emph{Mode Change Request} (MCR) is an event that triggers a mode change during a \emph{transition phase}.

In the context of RMDF, an MCR is the production of a control token by a mode decider. Following the criteria presented in~\cite{fort_synchronous_2022}, an MCP is classified according to three criteria: \emph{overlaping}, \emph{periodicity}, and \emph{retirement}. In our context, those criteria are defined as follows:

\begin{itemize}
  \item An \emph{overlapping} protocol allows actors from two distinct execution branches to be executed at the same time.
  \item A \emph{periodic} protocol allows actors outside conditional execution branches to not be interrupted by the mode change.
  \item A \emph{late-retirement} protocol allows actors from an old mode, i.e., a previous execution branch, to continue their execution for a given amount of time.
\end{itemize}

The MCP proposed in this paper for RMDF is non-overlapping, periodic, and late-retirement. In order to explain the non-overlapping property, let us remind restriction 5 of control areas (cf. \cref{prop:restrictions-control-areas}). This restriction states that whenever parametric rates of a control area are assigned a value, their sum is equal to 1. In conjunction with restriction 4, which states that production and consumption rates of a control area are equal to 0 or 1, it results in a single branch being executed at a time, resulting in the non-overlapping property. The proposed MCP is periodic as a mode change only affects the execution of actors in a control area. Finally, the late-retirement is ensured by the controlled joiner, which prevents the consumption of the token produced by the second branch until the first branch terminates. In other words, actors are allowed to finish their execution without compromising the expected order of tokens.

\section{Discussions and Related Works}

\label{sec:related-works}

\subsection{Comparison of RMDF with \textsc{Lustre} and \textsc{Prelude}}

Synchronous languages such as \textsc{Lustre}~\cite{halbwachs_synchronous_1991} and \textsc{Prelude}~\cite{forget_multiperiodic_2008} share the same goal as RMDF, that is to specify and analyze real-time systems. However, differences remain; the state of the art of \textsc{Lustre} and \textsc{Prelude} from~\cite{fort_programing_2022} and~\cite{forget_programming_2023} are used in the following paragraphs to explain those differences. An RMDF specification is composed of \emph{actors}, which communicate by exchanging data tokens on \emph{channels}, and a system specified with \textsc{Lustre} is composed of \emph{nodes}. Variables in nodes represent an infinite sequence of values, i.e., a \emph{dataflow} or, more simply, a \emph{flow}. Nodes are composed of \emph{equations}, which can be understood as mathematical equations: each equations define an output flow from input flow(s).

\textsc{Lustre} is structured around \emph{clocks}: each dataflow has a clock and is present only when the clock is present, i.e., \enquote{it produces a tick}. A static analysis named \emph{clock calculus} verifies the consistency of clocks, i.e., all combinations of flows can be evaluated without synchronization mechanisms such as buffering. This is different from RMDF (and dataflow models in general). In RMDF, produced data tokens are always present until they are consumed.

\textsc{Lustre} cannot directly specify real-time constraints, e.g., periods and deadlines. To overcome that, \textsc{Prelude} extends \textsc{Lustre} to allow system designers to explicitly define the period of a flow (i.e., the period of its clock) or let the compiler infer it. As an illustration, yet \emph{without proving any semantic equivalence}, the same logic in \textsc{RMDF} would be to specify the period of production and consumption of data tokens instead of the period of the actors.

A system specified with \textsc{Prelude} can then be composed of flows with different periods; those flows are not \emph{synchronous} in the sense of~\cite{halbwachs_synchronous_1991}. In order to combine non-synchronous flows, \emph{rate-transition operators} allow the production of a flow that is faster or slower than the input flow. As an illustration, yet again \emph{without proving any semantic equivalence}, the rational rates \textsc{RMDF} greater/lower than 1 allow for specifying faster/slower data tokens production and consumption.

\subsection{Comparing \textsc{RMDF} with an Extended \textsc{Prelude}}

The static analysis of the proposed RMDF relies on restrictions applied to control areas (cf. \cref{prop:restrictions-control-areas} for the exhaustive list of restrictions). Among those restrictions, the third asserts that \enquote{if there are timed actors in a control area, they must have the same frequency} and the fourth asserts that \enquote{the production rate inside a control area and the parametric rates associated to the controlled joiner and controlled splitter of that control area must be equal to one}.

\textsc{Lustre} has been extended in~\cite{colaço_conservative_2005} to allow the specification of multi-mode system. A switch-like statement allows to switch between sets of equations according to some enumerated value and an automaton construction. In addition, in the work of~\cite{colaço_conservative_2005}, all flows within a state must have the same period. The mode-dependent systems considered in this paper have a similar behavior because of the third and fourth restrictions presented in \cref{prop:restrictions-control-areas}.

The gap between and \textsc{Prelude} and~\cite{colaço_conservative_2005} is addressed in~\cite{fort_synchronous_2022}. The best of both works is used to support the specification of states with multi-periodic states in \textsc{Prelude}. The work of~\cite{fort_synchronous_2022} addressed the third and fourth restrictions of \cref{prop:restrictions-control-areas}. However, this work is in the context of synchronous languages, and work remains to be done to know if the same logic can be applied to RMDF and dataflow models.

\subsection{Dataflow Models with Conditional Execution}

Dynamic reconfiguration has been introduced in the dynamic PolyGraph model~\cite{dubrulle_dynamic_2019}. In this reference, a set of actors can select an execution mode. In addition to its data, each token is labeled with an execution mode. When an actor consumes tokens, it reads the mode of the token and decides how to process the data. Note that this processing occurs as a black box inside the actor's implementation. The processing of input tokens can also include discarding input tokens. A key advantage of this procedure is that the consistency and liveness checking algorithms of PolyGraph can be directly extended to dynamic PolyGraph. Mode execution is transmitted transitively from actor to actor, meaning all actors of a dynamic PolyGraph model must have a specific implementation to process execution mode. In the RMDF model, the execution modes are carried only from mode deciders to controlled splitters and controlled joiners. In other words, actors in a control area of an RMDF specification do not have a specific implementation that considers the execution mode. Therefore, RMDF allows for the easier integration of actors inside a control area. Our approach puts the complexity at the model level, while in dynamic PolyGraph, the complexity is at the implementation level.

TPDF~\cite{khanhdo_transaction_2016} is a generalization of CSDF where some actors are \emph{control actors}. Those latter are similar to our mode decider. A control actor outputs a control token, which a subsequent actor consumes to decide how this actor consumes or produces its tokens. For instance, it can select one of the input or output tokens or select the available input token with the highest priority. The expressivity of the different execution modes available is greater in TPDF than in RMDF. The key advantage of RMDF is that it is easier to specify data communication patterns and timing constraints than TPDF.

\section{Conclusion and Future Works}

\label{sec:conclusion-future-works}

We have presented the RMDF model, a dataflow model that extends the PolyGraph model. This latter is well-suited to specify CPSs with relaxed real-time constraints. Besides keeping this advantage, RMDF allows for the specification of CPSs with a mode-dependent execution, i.e., with a set of conditional execution branches. In such CPSs, some actors can select from/to which input/output channel tokens are consumed/produced, and this choice may depend on runtime data. We have extended the static analysis algorithm of the PolyGraph model to the RMDF model to verify the consistency and liveness of an RMDF specification. The timing behavior analysis of the PolyGraph model presented in~\cite{roumage_static_2024} is also extended to RMDF, and we presented a feasibility test. Static analyses of the RMDF model have also been implemented.

A dataflow model is a trade-off between expressiveness and analyzability. The expressiveness of RMDF could be improved by relaxing restrictions imposed on control areas, while the analyzability could be enhanced by incorporating mode dependencies within an RMDF specification. This approach is reasonable, as execution modes in real systems do not occur arbitrarily. More generally, it would be valuable to explore how the feasibility test proposed for RMDF can be applied to other design models beyond dataflow models, such as MatLab/Simulink. To that end, it is necessary to define semantically equivalent model transformations from those models to RMDF.

\bibliographystyle{IEEEtran}

\end{document}